\documentclass[
aps, 
superscriptaddress,
twocolumn, 
a4paper,
amsmath, amssymb, aps
]{revtex4-2}

\newtheorem{theorem}{Theorem}
\newtheorem{lemma}[theorem]{Lemma}
\newenvironment{proofT}[1]{\textit{Proof of {#1}.}}{\hfill$\square$}
\newenvironment{proof}[1]{\textit{Proof.}}{\hfill$\square$}

\usepackage{graphicx}
\usepackage[caption=false]{subfig}
\usepackage{hyperref}
\usepackage[ruled,vlined,linesnumbered]{algorithm2e}
\usepackage{physics}
\usepackage{enumitem}

\begin{document}

\title{Quantum speedup for track reconstruction in particle accelerators}

\author{D. Magano}
    \affiliation{Physics of Information and Quantum Technologies Group, Instituto de Telecomunica\c{c}\~{o}es, Portugal}
    \affiliation{Instituto Superior T\'{e}cnico, Universidade de Lisboa, Portugal}
\author{A. Kumar}
    \affiliation{Physics of Information and Quantum Technologies Group, Instituto de Telecomunica\c{c}\~{o}es, Portugal}
    \affiliation{Department of Mathematics, Clarkson University, USA}
\author{M. K\={a}lis}
    \affiliation{Center for Quantum Computer Science, Faculty of Computing, University of Latvia, Latvia}
\author{A. Loc\={a}ns}
    \affiliation{Center for Quantum Computer Science, Faculty of Computing, University of Latvia, Latvia}
\author{A. Glos}
    \affiliation{Institute of Theoretical and Applied Informatics, Polish Academy of Sciences, Poland}
\author{S. Pratapsi}
    \affiliation{Physics of Information and Quantum Technologies Group, Instituto de Telecomunica\c{c}\~{o}es, Portugal}
    \affiliation{Instituto Superior T\'{e}cnico, Universidade de Lisboa, Portugal}
\author{G. Quinta}
    \affiliation{Physics of Information and Quantum Technologies Group, Instituto de Telecomunica\c{c}\~{o}es, Portugal}
\author{M. Dimitrijevs}
    \affiliation{Center for Quantum Computer Science, Faculty of Computing, University of Latvia, Latvia}
\author{A. Rivo\v{s}s}
    \affiliation{Center for Quantum Computer Science, Faculty of Computing, University of Latvia, Latvia}
\author{P. Bargassa}
    \affiliation{Portuguese Quantum Institute, Portugal}
    \affiliation{Laborat\'{o}rio de Instrumenta\c{c}\~{a}o e F\'{i}sica Experimental de Part\'{i}culas, Portugal}
\author{J. Seixas} 
    \affiliation{Instituto Superior T\'{e}cnico, Universidade de Lisboa, Portugal}
    \affiliation{Portuguese Quantum Institute, Portugal}
    \affiliation{Center for Physics and Engineering of Advanced Materials, Portugal}
\author{A. Ambainis} 
    \affiliation{Center for Quantum Computer Science, Faculty of Computing, University of Latvia, Latvia}
\author{Y. Omar}   
    \affiliation{Physics of Information and Quantum Technologies Group, Instituto de Telecomunica\c{c}\~{o}es, Portugal}
    \affiliation{Instituto Superior T\'{e}cnico, Universidade de Lisboa, Portugal}
    \affiliation{Portuguese Quantum Institute, Portugal}

\date{\today}

\begin{abstract}
To investigate the fundamental nature of matter and its interactions, particles are accelerated to very high energies and collided inside detectors, producing a multitude of other particles that are scattered in all directions. 
As charged particles traverse the detector, they leave signals of their passage. 
The  problem  of  track  reconstruction  is  to  recover  the  original  trajectories  from  these signals. 
This challenging data analysis task will become even more demanding as the luminosity of future accelerators increases, leading to collision events with a more complex structure.
We identify four fundamental routines present in every local tracking method and analyse how they scale in the context of a standard tracking algorithm.
We show that for some of these routines we can reach a lower computational complexity with quantum search algorithms.
Although the found quantum speedups are mild, this constitutes, to the best of our knowledge, the first  rigorous evidence of a quantum advantage for a high-energy physics data processing task.
\end{abstract}

\maketitle

\section{Introduction}

Most of our understanding about fundamental interactions and the sub-nuclear structure of matter comes from exploring the results of colliding highly energetic particles in accelerator machines. 
These collisions produce a myriad of secondary particles, which must be detected and their trajectories subsequently reconstructed.
The search for new Physics beyond the Standard Model depends on being able to detect and process  extremely rare events among vasts amounts of data.
Experimental High-Energy Physics (HEP), especially the Large Hadron Collider (LHC) programme at the European Organization for Nuclear Research (CERN), is for this reason one of the most computationally demanding activities in the world \cite{Roadmap}. 
Moreover, this demand is expected to grow dramatically after 2026 with the upcoming High-Luminosity phase of the LHC \cite{HL-LHC}, and even more so in future machines, such as the Future Circular Collider \cite{FCC}.
As such, processing the data obtained in the particle detectors into useful information that can be analysed by high-energy physicists will become such a formidable task that it will likely require completely new technological paradigms.
Quantum computing, promising significant speedups or reduced computational and energetic resources for specific problems, may play a key role in overcoming these challenges.

In recent years, quantum solutions have been proposed for specific tasks in HEP data processing and analysis.
These include track reconstruction
\cite{TrackAnnealers,TrackAnnealingInspired,TrackClusteringAnnealers, Tuysuz2021,QMLinHEP, Quiroz2021}, event selection \cite{HiggsAnnealing, EventsQML,SauLanWu, qSVM21, Belis2021, GroverHEPevents, pires1, pires2}, and event simulation \cite{Geant2021}.
These promising proposals were typically conceived already with a quantum framework in mind, and were tested with very small problem instances due to the present quantum hardware limitations (for details on the computational scaling of previous quantum algorithms for track reconstruction see Appendix \ref{app:comparison}).
Therefore, it remained an open question whether one could prove a quantum speedup for a relevant task meeting the large-scale requirements of modern HEP data processing.
One route towards this goal is to consider the computational complexity of standard HEP algorithms and whether quantum computers could be used to improve it.
In \cite{Jets}, for example, the classical and quantum computational scaling of a well-known jet clustering algorithm is studied: a quantum algorithm with speedup is found, as well as an alternative classical algorithm that matches the quantum performance, therefore establishing no quantum advantage.

In this article, we consider the problem of track reconstruction (also known as \emph{tracking}) from a computational complexity perspective.
In particle physics experiments, bunches of accelerated particles are collided inside tracking detectors.
At these collisions, new particles are created and scattered in all directions.
As charged particles cross the detector's multiple layers, they leave signals of their passage, which are converted into three-dimensional points called \emph{hits}.
The collection of hits that are left by such a particle is called that particle's \emph{track}.
The event record of an experiment consists of the totality of signals from all particles of an interaction (or possibly  several interactions) after one full readout of the detector.
The goal of tracking is to reconstruct the particles' tracks from the event record --
see Figure \ref{fig:event_example} for an illustration of the problem.
Given that in real experiments an event can contain several thousand hits, most combinations of hits (track candidates) will not correspond to an actual particle.
Therefore, we need efficient algorithms to be able to reconstruct the tracks in a reasonable time.

\begin{figure*}[t]
\centering
\subfloat[Input.]{\label{fig:input}
  \includegraphics[width=0.35\linewidth]{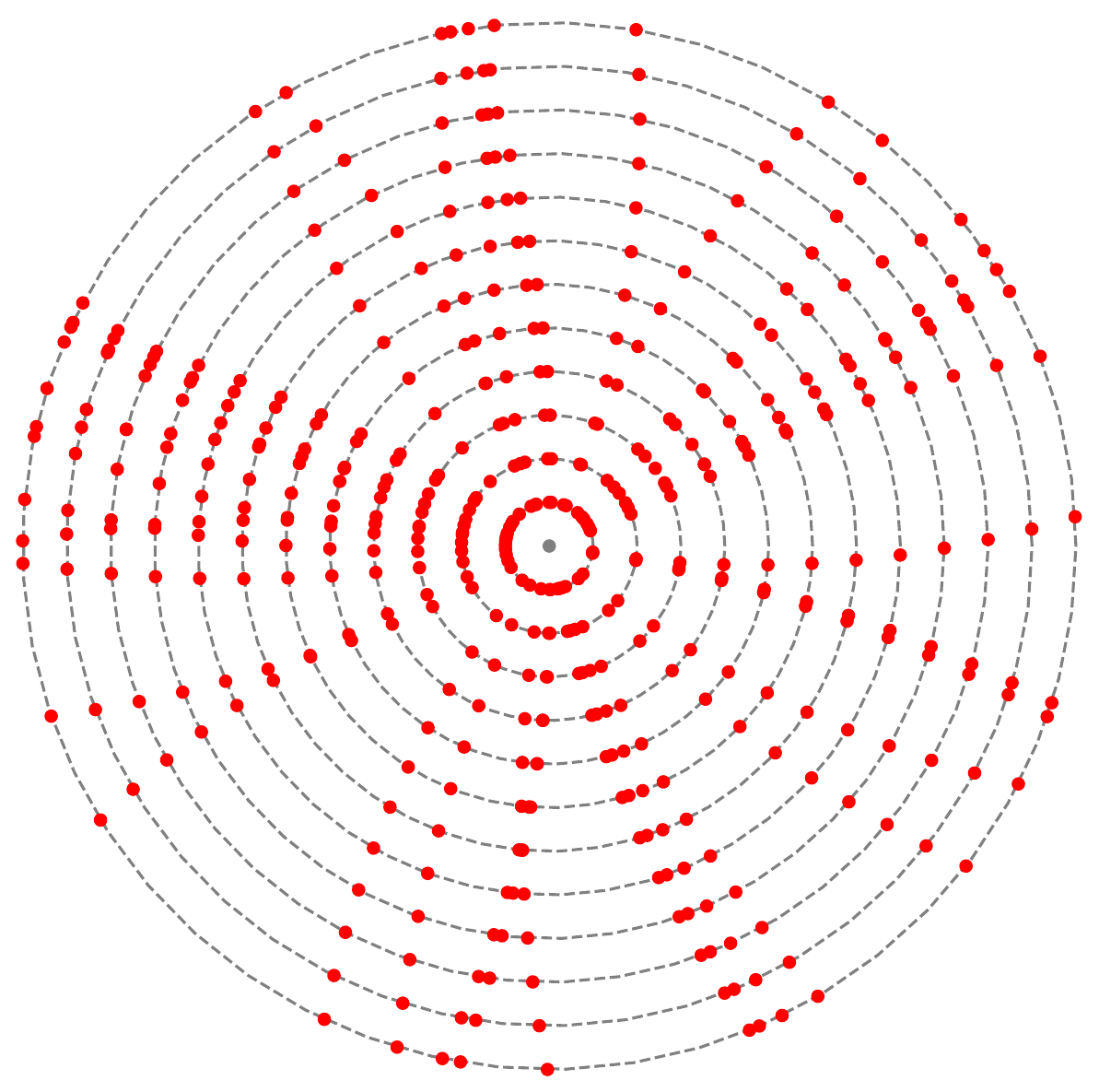}
}
\qquad
\subfloat[Output.]{\label{fig:output}
  \includegraphics[width=0.35\linewidth]{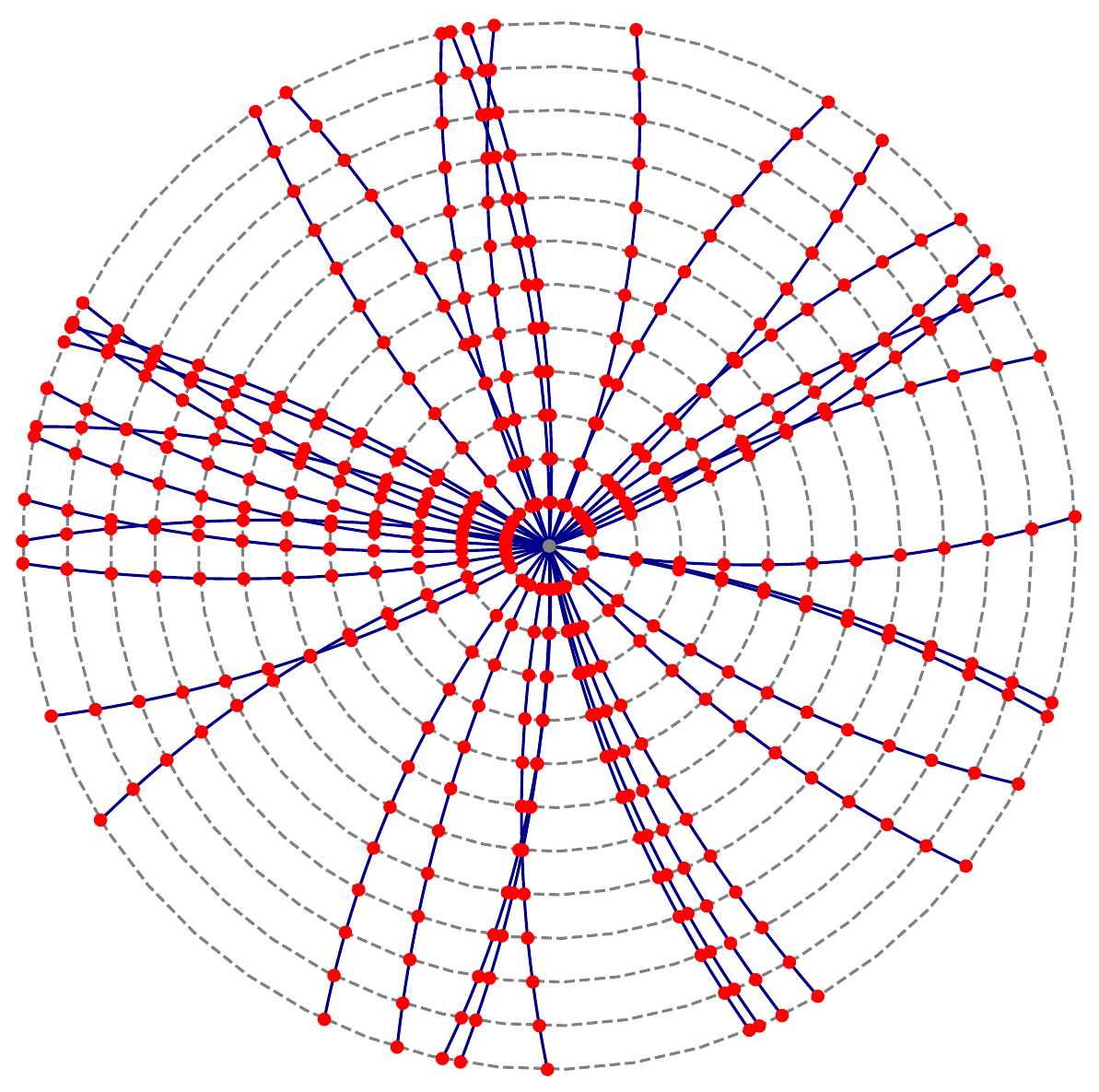}
}
\caption{Illustration of track reconstruction.
Transverse view  (with respect to the beam line) of a tracking detector with cylindrical  layers  (dashed  grey  lines).
The input to tracking is a set of hits (red circles) corresponding to detections of the particles' passage (a).
We recover the original trajectories (blue lines) by grouping hits that belong to the same particle, i.e., by reconstructing the particles' tracks (b).}
\label{fig:event_example}
\end{figure*}

The current tracking methods can be broadly classified into local and global methods \cite{ReviewTracking}.
Global methods treat all hit information in an equal and unbiased way and are essentially clustering algorithms in some feature space.
All the quantum approaches so far were based on global methods.
However, because these methods can be very inefficient in terms of speed, local methods are still the standard at several reconstruction programmes in high-energy physics \cite{JINST, ATLAS, Braun2018}.
For this reason, they will be the focus of this article.

We identify four fundamental computational routines that are present in every local tracking method: seeding, track building, cleaning, and selection.  
In the first stage, seeding, we form initial rudimentary track candidates, called \emph{seeds}, using just a few hits.
Then, in the track building stage, we extrapolate the seeds’ trajectories along the expected path of the particle and build track candidates by adding compatible hits from successive detector layers.
This strategy may lead to multiple track candidates describing the same particle.
To avoid such redundancies, we apply a cleaning process that removes track candidates that are too similar.
Finally, only the track candidates that respect some quality criteria (based on the quality of the fit between the trajectory and the corresponding hits) are output from the reconstruction process -- this is called the selection stage.
Figure \ref{fig:tracking} provides a summary of the four stages.

For each of these stages, we analyze the computational scaling of the Combinatorial Track Finder (CTF) algorithm \cite{JINST}, which was the basis of the tracking program of the CMS collaboration during the 2016 LHC  run \cite{Sguazzoni2016}.
While we focus on the CTF for concreteness, we point out that the underlying structure is the same as for most local track reconstruction methods, such as the ones used in ATLAS \cite{ATLAS} or Belle II \cite{Braun2018}.
For both the seeding and track building stages, we show that we can reproduce the same output (up to bounded error probability) with lower quantum complexity by an adequate use of quantum search routines.
For the cleaning routine, we find an alternative classical algorithm with improved scaling that is optimal up to poly-logarithmic factors.
The selection stage is already trivially optimal from a complexity perspective.
We emphasize that the four tracking routines are analyzed independently, adding flexibility to our results.
For example, the CMS collaboration recently adopted a different seeding strategy \cite{Bocci2020}, but the structure of the other three stages remains  unchanged.
Finally, we consider executing the entire reconstruction coherently, where we do not register the outputs of the individual stages, but are only interested in the final reconstructed tracks.
We show that this scenario can lead to further quantum advantage.
Our results are summarized at the end of the article in Table \ref{tab:summary}.

\begin{figure*}[t]
\centering
\subfloat[Seeding.]{\label{fig:seeding}
  \includegraphics[width=0.49\linewidth]{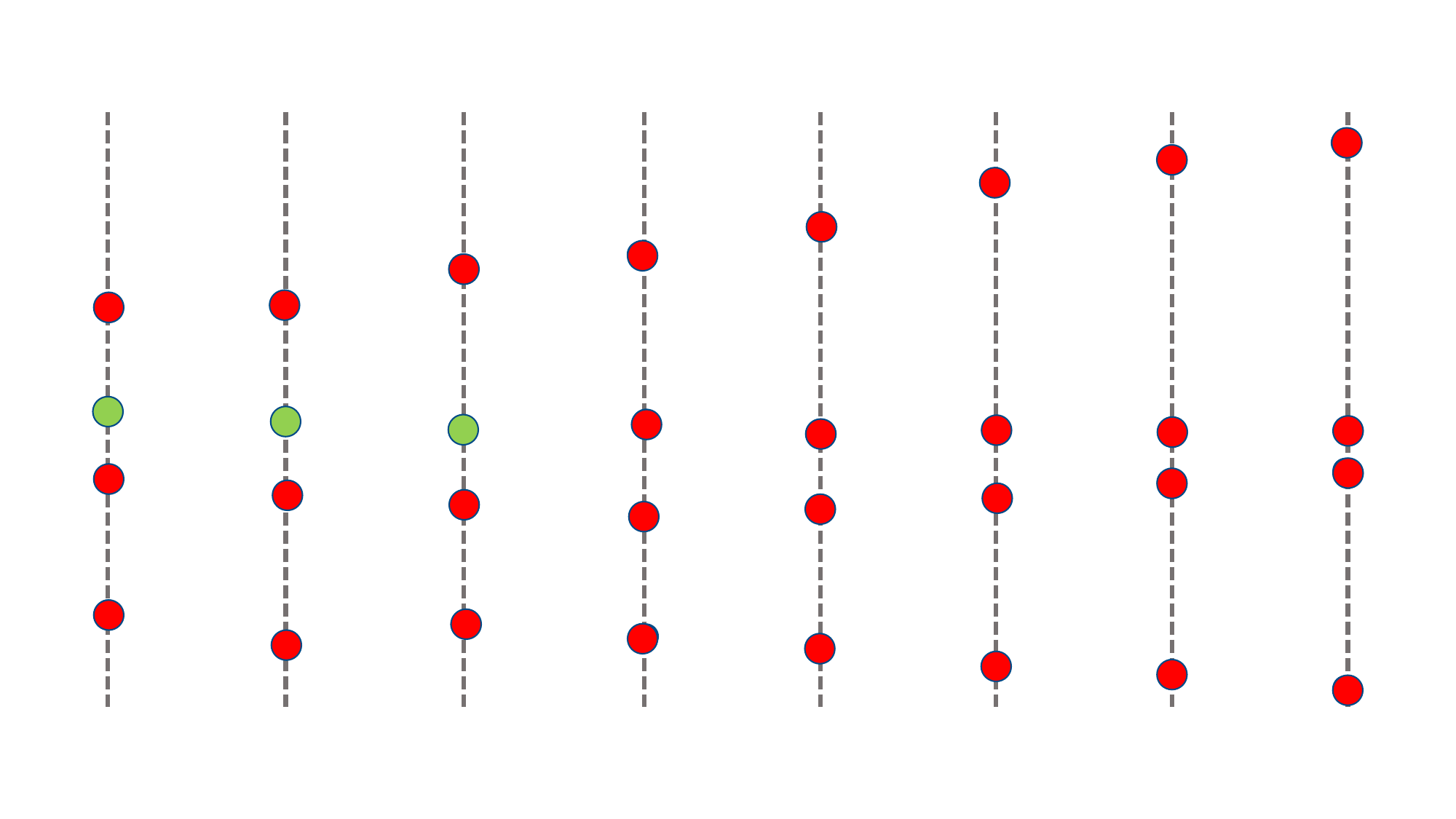}
}
\subfloat[Track Building.]{\label{fig:building}
  \includegraphics[width=0.49\linewidth]{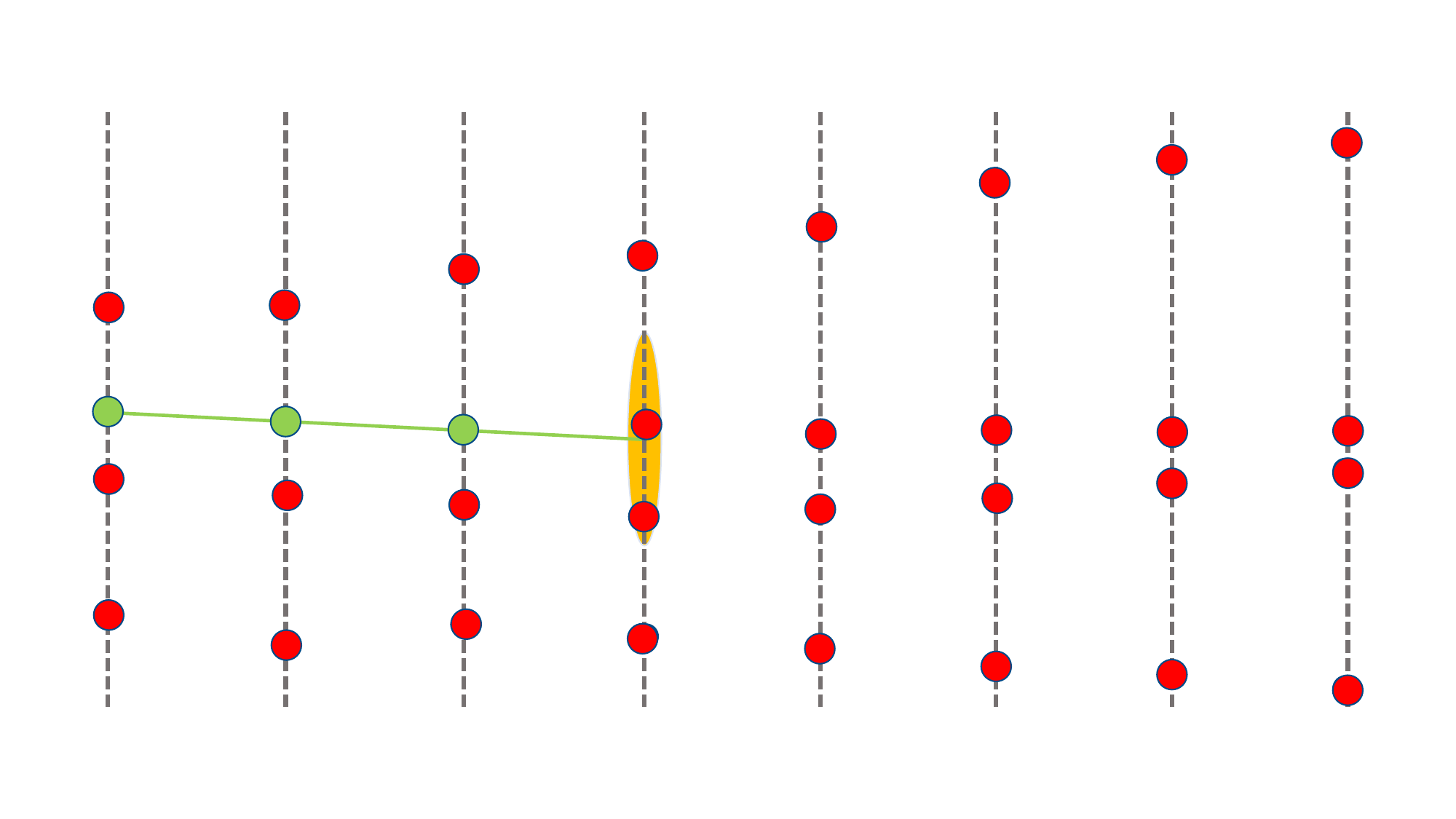}
}
\hfill
\subfloat[Cleaning.]{\label{fig:cleaning}
  \includegraphics[width=0.49\linewidth]{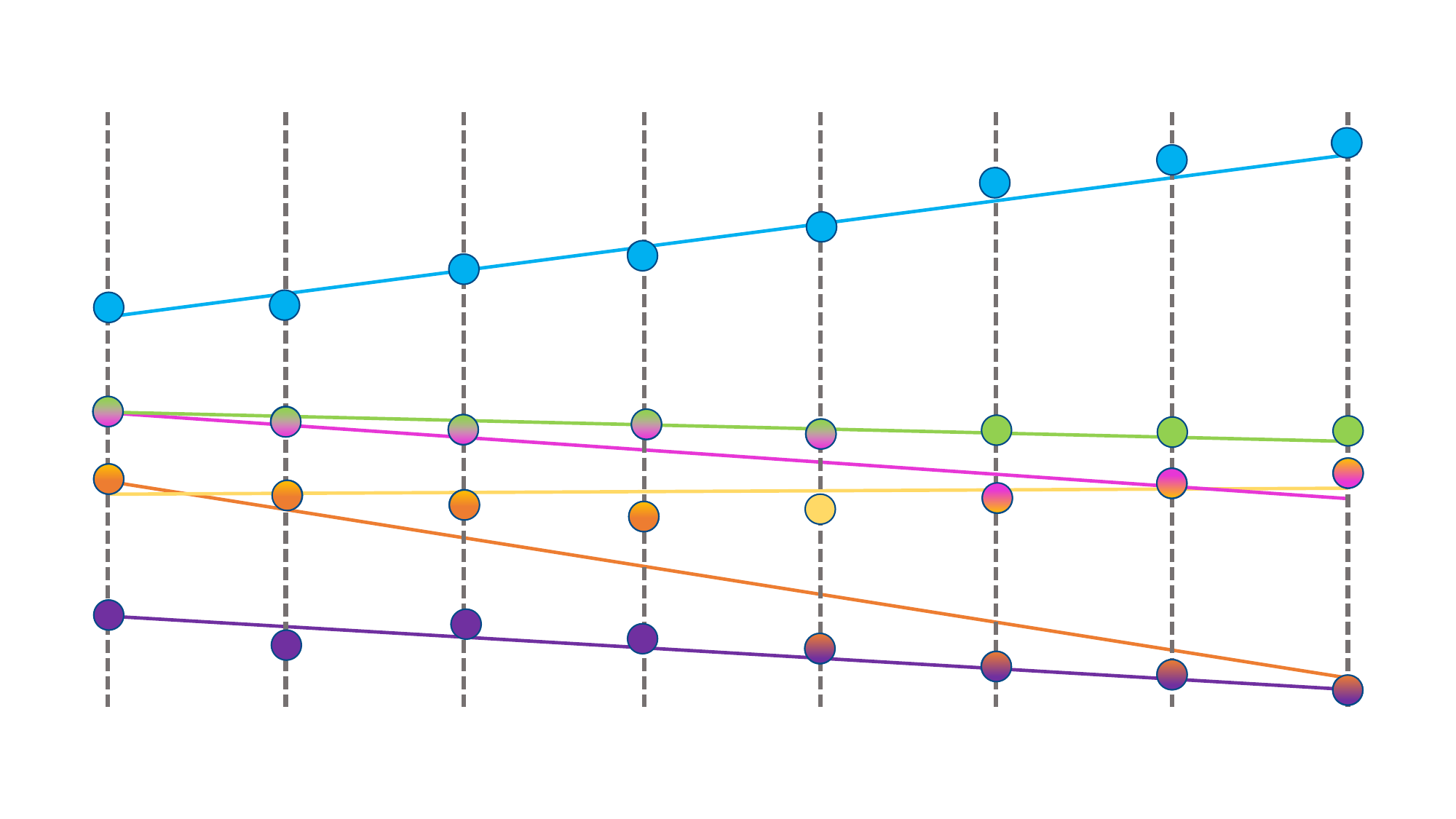}
}
\subfloat[Selection.]{\label{fig:selection}
  \includegraphics[width=0.49\linewidth]{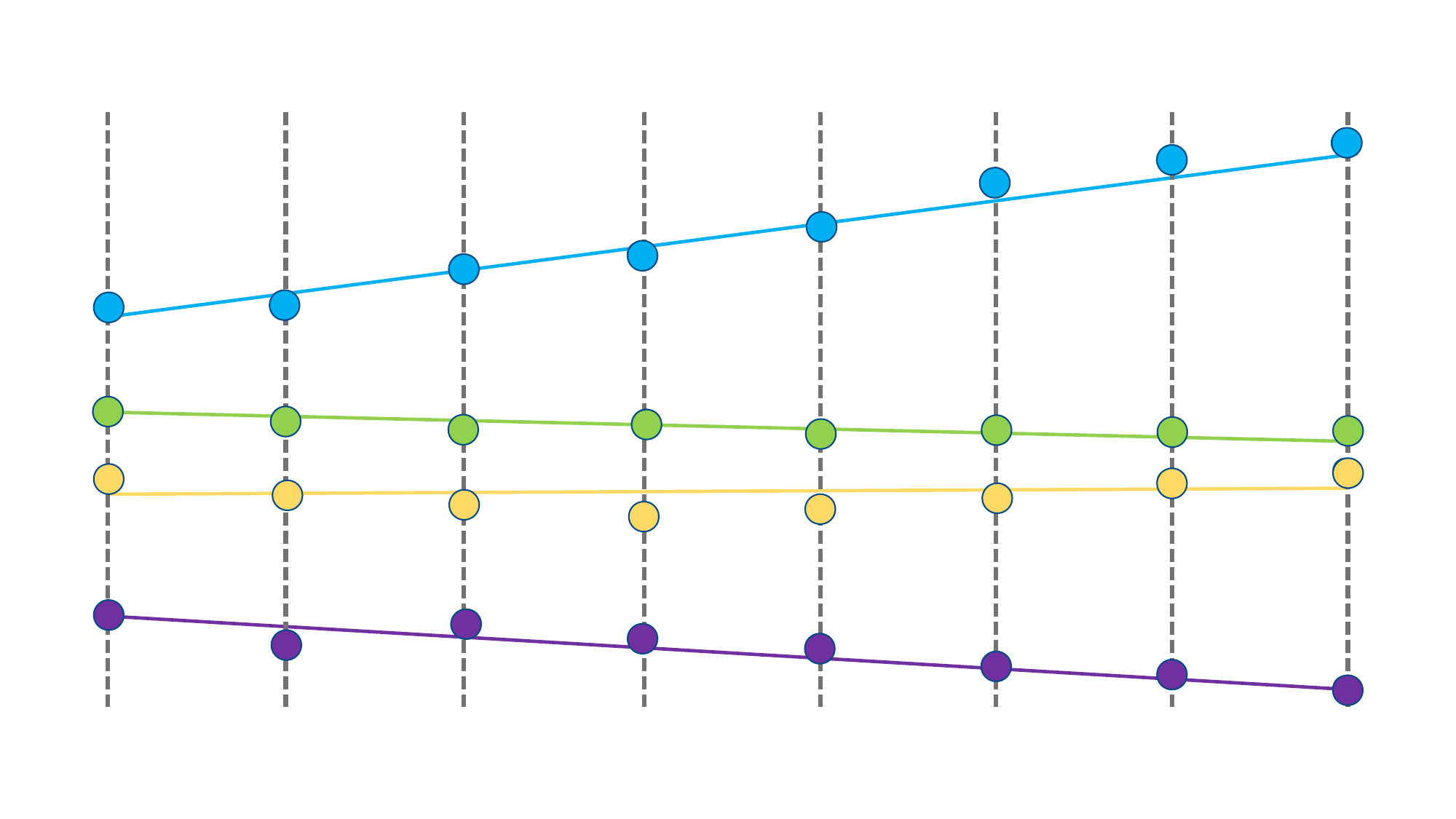}
}
\caption{Track reconstruction: the four stages.
In these figures, the charged particles resulting from a collision travel from left to right, recording hits (circles) as they cross the detector layers (dotted grey lines).
(a) Seeding.
Local track reconstruction methods start by forming seeds, rudimentary track candidates with just a few hits.
In this figure we highlight in green a possible seed with three hits.
(b) Track Building.
The trajectory of each seed is extrapolated throughout the detector, building track candidates by adding compatible hits layer by layer.
(c) Cleaning.
The found tracks are here represented in different colours.
During the cleaning stage, if two track candidates share too many hits one of them gets discarded.
The pink track, for example, is filtered at this stage because of the hits it shares with the green one.
(d) Selection.
Only the track candidates that satisfy certain quality criteria are selected as the output of the reconstruction process.}
\label{fig:tracking}
\end{figure*} 

\section{The tracking problem}

We present a simplified model of tracking, omitting some details that do not significantly influence the complexity analysis.
In Appendix \ref{app:trackingmodel} we formalize further our assumptions and discuss how they could be relaxed.

Let $n$ denote the number of charged particles present in the event record.
Their trajectories originate from a fixed interaction region, but not necessarily from a single collision spot  (this accounts for the \emph{pileup} effect \cite{JINST}).
Since the detector is immersed in a quasi-uniform co-axial magnetic field,  we expect the particles to follow helical trajectories aligned with the field's direction.
We model the particles' trajectories as independent and identically distributed random variables, such that the underlying probability distribution is strictly positive on the corresponding parameter space.
The assumption that the different trajectories are uncorrelated represents a common practice in track reconstruction.
The detector consists of a set of $L$ cylindrical sensor layers aligned with the beam line.
The layers are indexed from $0$ to $L-1$ (the most inner layers having the smallest indices). 
We assume that each particle traverses every layer exactly once.
As a particle crosses a layer, it leaves a complex signal resulting from the interaction with the sensor's pixels.
We treat these signals simply as three-dimensional points, called hits.
Thus, we require the granularity of the sensors to be high enough such that each detected hit can always be differentiated from others. 
At every layer, we identify the hits with labels from $\{0, \ldots, n-1\}$, using the notation $\mathbf{m}_{l,j}$ for the coordinates of $j$th hit in layer $l$.
It is possible that some hits are not measured at all due to sensor inefficiencies, that is, we do not necessarily have a hit $\mathbf{m}_{l,j}$ for every pair $(l,j)$.

\section{Computational Model \label{sec:computationalmodel}}

The running time of any tracking algorithm is the result of different factors.
Naturally, the larger the number of recorded particles the more demanding track reconstruction becomes.
At the LHC, sophisticated computational architectures are employed to optimize the running time \cite{ICTchallenges, PantaleoThesis}.
Furthermore, the parameters of the tracking software (such as the minimum $p_T$ of the tracks, see Appendix \ref{app:seedingalgorithm})are carefully adjusted to achieve in useful time a track reconstruction with the desired accuracy, under some assumptions on the observed events.

In this work, we offer a different perspective on tracking, focusing on the computational complexity of the problem.
In other words, we are interested in understanding how it fundamentally scales with input size.
As is common in the theoretical analysis of algorithms, we concern ourselves with the asymptotic limit of arbitrarily large number of particles.
We adopt the standard ``big O'' notation for asymptotic upper bounds. For two functions $f$ and $g$ from $\mathbb{N}$ to $\mathbb{R}$ we say that $f=O(g)$ if $\exists C, x_0 > 0: \forall x, \left(x > x_0 \implies f(x) < C \cdot g(x) \right)$. 
We write $f=\Omega(g)$ if $g=O(f)$.
We say that $f=\Theta(g)$ if $f=O(g)$ and $g=O(f)$. By ``constant time'', we mean $O(1)$.

For our complexity analysis we only consider the dependence on the variable $n$, the number of particles.
Evidently, the data in the event record also depends on quantities like the number of layers of the detector, the granularity of the sensors, or the efficiency of the detectors.
But these are fixed from the experimental hardware and do not vary from event to event.
On the other hand, we expect the average $n$ to grow as we increase the beam's instantaneous luminosity.
Thus, we believe $n$ to be an appropriate measure of the size of the input to the tracking problem.

When considering the classical algorithms we assume that, given $(j,l)$, we can access $\mathbf{m}_{l,j}$ in constant time.
Moreover, simple arithmetic operations on the hit's coordinates are counted as taking $O(1)$ time.
In the context of the quantum algorithms, we work in the circuit model, measuring time as the number of quantum gates used. 
We assume access to a QRAM that is able to load classical data in coherent superposition in logarithmic time in the number of memory cells \cite{QRAM}.
That is, we have access to a unitary $\mathbf{Q}$ such that, given a superposition $\ket{\psi} = \sum_{l, j} \alpha_{l, j} \ket{l,  j } \ket{ 0 }$, applying $\mathbf{Q}$ yields the state
\begin{equation}
\mathbf{Q} \ket{\psi}
=
\sum_{l, j} \alpha_{l, j} \ket{l,j} \ket{b} \ket{\mathbf{m}_{l,j} }.
\end{equation}
In this expression, $\ket{b}$ is a flag qubit indicating whether the index $j$ corresponds to an actual hit in layer $l$ (there may be fewer than $n$ detections per layer).
If so, $\ket{\mathbf{m}_{l,j}}$ is a computational basis quantum state encoding the coordinates  of the $j$th hit of layer $l$ (otherwise, this register can be in an arbitrary state, say, the all-zero state).

Our conclusions will be critically dependent on the existence of a QRAM with the above mentioned properties, which represents a common practice in theoretical work on quantum algorithms.
Nevertheless, we point out that, even though there have been proposals of physical architectures for implementing QRAM \cite{QRAM_architectures}, there are still significant challenges to overcome before such a device can be realized in practice.

Our choices of computational models represent the distinct standard practices in classical/quantum algorithm analysis.
To attenuate the differences in the computational models, we present the results in $\tilde{O}$ notation, that is, omitting the poly-logarithmic dependencies in the complexities.

Another difference between the classical and quantum scenarios is that all the presented classical algorithms are deterministic, meaning that for a given input they will always output the same answer.
On the other hand, our quantum algorithms are probabilistic.
That is, they output the correct answer with some constant probability.
This probability can always be amplified to $1-\epsilon$, for any $\epsilon \in [0,1[$, at a cost of $O(\log(1 / \epsilon))$ repetitions.

\section{Four stages in track reconstruction}

\subsection{On the density of hits}

Naturally, the number of hits in the event record grows proportionally to the number of particles.
In fact, for any fixed region of the detector we expect the number of registered hits to grow linearly with $n$.
In Appendix \ref{app:trackingmodel} we rigorously establish the following Lemma, which will become useful for the complexity analysis:

\begin{lemma}
For almost all events, the number of hits in any fixed, open (non-empty) subset of any detector layer grows as $\Theta(n)$.
\label{thm:almostallevents}
\end{lemma}

\subsection{Seeding}

Local track reconstruction methods start by forming rudimentary track candidates, known as track seeds, with just a small number of hits from a specific region of the detector (usually the innermost layers due to their higher granularities).
In the seeding scheme of \cite{JINST}, seeds are formed with only three hits.
But our analysis easily extends to seeds with $c$ hits, so in what follows we consider this more general scenario.
We impose some constraints on the trajectories defined by the seeds, namely a minimum transverse momentum and a maximum transverse and longitudinal distances to the presumed production point of the particle.
Any hit $c$-tuplet disobeying these conditions is not considered a valid seed.

The CTF's seeding routine first searches over the first two seeding layers of the detector for pairs of hits that are compatible with the seeding criteria.
These pairs are extended into hit triplets by searching over the third layer for compatible hits (meaning that we want the hit triplets to respect the imposed conditions on the seeds' trajectories).
These triplets are extended in a similar way into hit quadruplets, then hit $5$-tuplets, and so on until the $c$-th layer is reached.
In Appendix \ref{app:pseudocodes} we provide a pseudo-code representation of the seeding routine (Algorithm \ref{algo:seeding}).

Suppose, now, that we have carried this process until the $j$-th layer (for some general $j$ between 2 and $c-1$).
Each hit $j$-tuplet determines a region in the $(j+1)$-th layer where we could find a hit continuation compatible with the seeding constraints -- see Figure \ref{fig:seedsearching} for an illustration of  the $j=2$ case.
According to Lemma \ref{thm:almostallevents}, we expect to find $\Theta(n)$ hits in that region.
So, the number of selected tuplets grows by a multiplicative $\Theta(n)$ factor at each layer (although the multiplicative constant shrinks as $j$ increases), resulting in $\Theta(n^c)$ seeds.

\begin{figure}[t]
\centering
\includegraphics[width=0.9\linewidth]{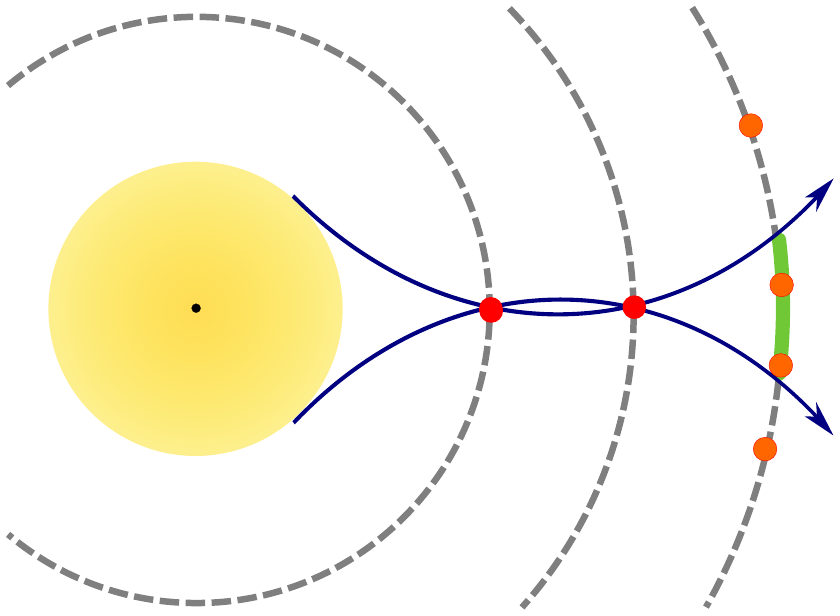}
\caption{Searching for seeds.
The dotted grey lines represent the first three detector layers (transverse view).
The yellow region around the beam axis (the black dot) is the region where we admit the collisions may occur.
In this illustration, we are forming a seed with the two hits from the inner layers marked in red.
For any hit in the third layer (orange circles) to be considered compatible with this seed, it must lie in the green region.
This is the region where the trajectories that respect the seeding criteria and that pass through the two red hits cross.
In blue we draw two trajectories originating from the outer edge of the collision region.}
\label{fig:seedsearching}
\end{figure}

\begin{theorem}
The CTF's seeding algorithm, Algorithm \ref{algo:seeding} (classical), has complexity $O(n^c)$, where $c$ is the number of hits per seed.
\label{thm:classicalseeding}
\end{theorem}

In practice, the number of seeds may be much smaller than $n^c$.
However, the result of Theorem 1 is independent of how many seeds are actually formed.
In contrast, it is simple to see that quantum computers can reach a lower complexity if $k_{\text{seed}}$, the number of seeds, scales better than $O(n^c)$.
Indeed, let $\mathbf{O}_{\text{seed}}$ be an operator that, given a state $\ket{j_0, \ldots, j_{c-1}}$, applies a $-1$ phase if $(\mathbf{m}_{0,j_0}, \ldots, \mathbf{m}_{c-1,j_{c-1}})$ corresponds to a valid seed, and leaves it unmarked otherwise.
Within the QRAM model presented in Section \ref{sec:computationalmodel}, we can implement $\mathbf{O}_{\text{seed}}$ with a $\tilde{O}(1)$-sized circuit.
Then, the idea is to use Grover's quantum search algorithm \cite{Grover1996, BoyerSearch} to find the good seeds.
Let
\begin{equation}
\mathbf{G} := \mathbf{H} \cdot \left( 2 \ket{0}\bra{0} - I \right) \cdot \textbf{H} \cdot \mathbf{O}_{seed},
\label{eq:Gdefinition}
\end{equation}
where $\textbf{H}$ is the Walsh-Hadamard transform \cite{NielsenChuang}.
Starting with a uniform superposition over the hit $c$-tuplets, we keep apply $\mathbf{G}$ 
\begin{equation}
O\left(\sqrt{\frac{n^c}{k_{\text{seed}}}}\right)
\end{equation}
times to amplify the probability of sampling a good seed to $\Omega(1)$.
Repeating this $\tilde{O}(k_{\text{seed}})$ times suffices to sample all good seeds with high probability.
We note that this approach works even if $k_{\text{seed}}$ is \emph{a priori} unknown since it can be determined with quantum counting \cite{quantumcounting}.
The detailed steps are presented in Algorithm \ref{algo:quantumseeding}.

\begin{theorem}
Algorithm \ref{algo:quantumseeding} (quantum) generates all seeds with bounded-error probability in  expected time $\tilde{O} \left( \sqrt{k_{\mathrm{seed}} \cdot n^c} \right)$, where $k_{\mathrm{seed}}$ is the total number of generated seeds.
\label{thm:quantumseeding}
\end{theorem}

For the proofs Theorems \ref{thm:classicalseeding} and \ref{thm:quantumseeding} we refer the reader to Appendix \ref{app:seedingalgorithm}.

\subsection{Track building \label{sec:trackbuilding}}

The CTF's track building phase is based on the combinatorial Kalman filter \cite{KalmanFilterFruhwirth} method, which is an adaptation of the Kalman filter \cite{Kalman} for tracking problems.
We build track candidates starting from the seeds by adding new hits, layer by layer, until the end of the detector is reached.
The idea is that, at each step, we extrapolate the candidate trajectory until it intersects the subsequent layer.
Then, to every hit on that layer we attribute a $\chi^2$ value, which essentially translates the distance between the hit and the intersection point of the trajectory with the layer (for detailed definition see Appendix \ref{app:trackbuilding}).
We add the hits with lowest $\chi^2$ value to the track candidate and then update the trajectory's parameters  according to the Kalman filter method.

Crucially, at each propagation step we may branch each track candidate into new candidates if several continuation hits are consistent with the present knowledge of track parameters.
The possibility that the expected hit is simply missing (for example, due to device inefficiencies) also gives rise to a new branch with no hit added at that layer -- this is known as a ``ghost hit''.
Following an argument similar to that of the previous Section (relying on Lemma \ref{thm:almostallevents}), we would expect $\Theta(n)$ new branches to be formed per layer per track candidate. 
However, to prevent a rapid increase in the number of tracks the CTF algorithm imposes a limit of five candidates retained at each step per starting seed.
These five candidates are selected based on their quality score, which is the number of hits added to the track (excluding ghost hits) minus a scaled sum of $\chi^2$ values for each of the hits.

In summary, for each track candidate, the propagation step involves: extrapolating the trajectory onto the next layer ($O(1)$ time), selecting the five best branches ($O(n)$ time), and updating the trajectory parameters for each branch (also $O(1)$ time). 
For more details on the track building stage, we refer the reader to Appendix \ref{app:trackbuilding}.
For a pseudo-code of the track building stage see Algorithm \ref{algo:finding}.

\begin{theorem}
Starting from $k_{\mathrm{seed}}$ seeds, the CTF's track building algorithm, Algorithm \ref{algo:finding} (classical), has complexity $O(k_{\mathrm{seed}} \cdot n)$.
\label{thm:classicalbuilding}
\end{theorem}

With access to a quantum computer, at each propagation step we can prepare all new branches in superposition. 
We propose using D\"{u}rr and H{\o}yer's quantum minimum finding algorithm \cite{DurrMinimum} to select the best five ones in $\tilde{O}(\sqrt{n})$ time, circumventing the classical cost of having to inspect the $O(n)$ branches one-by-one.
In more detail, suppose that we have followed a track candidate up to layer $l-1$. 
Calling the QRAM, we can build a quantum circuit $\mathbf{O}_{\text{find}}$ that, given a state $\ket{l,j} \ket{q}$, applies a $-1$ phase if adding the hit $(j,l)$ yields a new track candidate with $\chi^2$ value smaller than  $y$.
We then perform quantum search with $\mathbf{O}_{\text{find}}$ as the oracle, setting some initial $y \leftarrow y_0$.
If we find a hit with corresponding $\chi^2$ value $y' < y_0$, we run the quantum search again with $y \leftarrow y'$ (see Algorithm \ref{algo:quantumminimumfinding} for the detailed steps).
With high probability, in $\tilde{O}(\sqrt{n})$ this procedure will have converged to the branch with the lowest $\chi^2$ value.

\begin{theorem}
Starting from $k_{\mathrm{seed}}$ seeds, Algorithm \ref{algo:quantumfinding} (quantum) performs track building with bounded-error probability in $\tilde{O} \left( k_{\mathrm{seed}} \cdot \sqrt{n} \right)$ time.
\label{thm:quantumbuilding}
\end{theorem}

The full proofs of Theorems \ref{thm:classicalbuilding} and \ref{thm:quantumbuilding} are provided in Appendix \ref{app:trackbuilding}.

\subsection{Cleaning}

The combinatorial Kalman filter method may yield multiple tracks corresponding to the same particle,  by either starting from different seeds, or when a seed grows into more than one track.
To avoid this, the cleaning stage calculates the fraction of shared hits between all pairs of track candidates
\begin{equation}
\frac{N^{\text{hits}}_{\text{shared}}}{\min \left( N^{\text{hits}}_1, N^{\text{hits}}_2 \right)},
\end{equation}
where $N^{\text{hits}}_1$ ($N^{\text{hits}}_2$) is the number of hits used in forming the first (second) track and $N^{\text{hits}}_{\text{shared}}$ is the number shared hits between the two tracks.
If for any pair this fraction exceeds a fixed threshold value, the worst track  (i.e., the one with the lowest quality score) gets discarded  -- see Algorithm \ref{algo:cleaning}.
This pairwise comparison method leads to a quadratic scaling with the number of track candidates.

\begin{theorem}
CTF's cleaning algorithm, Algorithm \ref{algo:cleaning} (classical), has complexity $O \left( k_{\mathrm{cand}}^2 \right)$, where $k_{\mathrm{cand}}$ is the number of track candidates in the input.
\label{thm:classicalcleaning}
\end{theorem}

Now consider, for simplicity, that all tracks have the same number of hits, say $L$.
We can find an  asymptotically more efficient classical algorithm by observing that (a) there is an integer $r$ (independent of $n$) such that 
two tracks exceed the allowed fraction of shared hits if and only if 
they have $r$ hits in common and (b) each track only has $\binom{L}{r}=O(1)$ distinct $r$-tuples of hits.
We start by sorting the candidate tracks by quality score, such that if we need to discard one of two tracks we choose the one that is further down the list.
Evidently, the first track $t_1$ is going to be included in the output.
We create a self-balancing binary search tree $\mathcal{T}$, like a red-black tree (\cite{IntroAlgs}, for example), containing all of the $r$-tuples of hits of $t_1$ (with some induced order on the $r$-tuples).
We then move to the second track in the list $t_2$.
For every $r$-tuple of  $t_2$, we search for a match in the tree $\mathcal{T}$.
If we do not find any, we insert all of $t_2$'s $r$-tuples into $\mathcal{T}$ and we add $t_2$ to the output.
Otherwise, $t_2$ is not included in the output and we leave the tree unchanged.
We repeat this procedure for the remaining tracks.
In the end, the output contains all the desired tracks.
In Appendix \ref{app:cleaning} we extend this method to the case of general track sizes, resulting in Algorithm \ref{algo:improvedcleaning}, and prove

\begin{theorem}
Algorithm \ref{algo:improvedcleaning} (classical) performs track cleaning in $\tilde{O} \left( k_{\mathrm{cand}} \right)$ time, where $k_{\mathrm{cand}}$ is the number of track candidates in the input.
\label{thm:newcleaning}
\end{theorem}

\subsection{Selection}

The selection routine filters out the track candidates whose quality score falls bellow a specified threshold -- see Algorithm \ref{algo:selection}.
Since the quality score of any track candidate is independent of all the other tracks, this stage clearly scales linearly with the number of track candidates.

\begin{theorem}
CTF's selection algorithm, Algorithm \ref{algo:selection} (classical), has complexity $O(k_{\mathrm{cand}})$, where $k_{\mathrm{cand}}$ is the number of track candidates in the input.
\label{thm:selection}
\end{theorem}

\section{Reconstructing Tracks in Superposition}

We have seen that the seeding stage may exhibit quantum speedup if the number of seeds, $k_{\text{seed}}$, is considerably smaller than the total number of combinations of $c$-tuplets, $n^c$ (Theorem \ref{thm:quantumseeding}).
However, in general, the number of seeds will scale like $\Theta(n^c)$, and the quantum complexity will be the same as the classical one. 
In that case, only the track building stage shows a proven lower complexity: $O(n^{c+1})$ classical (Theorem \ref{thm:classicalbuilding}) versus $O(n^{c + 0.5})$ quantum (Theorem \ref{thm:quantumbuilding}).
If the four stages are run sequentially, the track building stage will dominate the CTF's complexity (both in the classical and quantum cases).

Now suppose that we are only interested in the final reconstructed tracks.
Instead of producing the output of each stage before continuing with the next one, we propose an algorithm, relying on quantum superposition over all track candidates, that reconstructs the full tracks one-by-one.
This further improves the quantum advantage provided that the number of reconstructed tracks is $O(n)$.
Intuitively, this condition means that the CTF can be applied in the asymptotic regime while keeping a constant fraction of tracks that do not correspond to a real charged particle.
We point out that, in practice, particle physicists empirically adjust the parameters of the tracking software according to the luminosity regime to obtain a reasonable fake track rate for most events. Alternatively, one may think that we are only interested in reconstructing the best $O(n)$ tracks.

The promise that only $O(n)$ tracks are to be found among $\Theta(n^c)$ track candidates suggests the use of quantum search, as we did with seeding. 
This is complicated for two reasons: (a)~the track building routine \emph{forgets} information by selecting only some track candidates in each layer, i.e., it is not reversible, while quantum search relies on (reversible) unitary transformations; (b)~the cleaning operation for each track candidate depends on information about other tracks. 
To rectify point (a) we apply the principle of deferred measurements \cite{NielsenChuang} to create a sequence of unitary transformations that mimic the CTF algorithm. 
To rectify point~(b) we adapt our improved cleaning algorithm, coherently accessing the nodes of  the search tree via QRAM.

More concretely, after reconstructing the first $i$ tracks, we build a circuit $\mathbf{U}_i$ that prepares a superposition of all $\tilde{O}(n^c)$ fully built tracks, flagging the ones that have already been accepted as reconstructed tracks.
This circuit only requires $\tilde{O}(\log n^c \cdot \sqrt{n})$ gates using our quantum routine for track building (Section \ref{sec:trackbuilding}).
Then, we sample the best-scoring track candidate in that superposition that does not overlap with any previously reconstructed track candidates to form the $(i+1)$th track in the output.
Using quantum minimum finding, this can be done with $O(\sqrt{n^c})$ calls to  $\mathbf{U}_i$.
We repeat this procedure until no new valid tracks are found -- $O(n)$ times due to the promise.
Our proposal is summarized in Algorithm \ref{algo:reco_in_superposition}.
For a detailed construction of the unitaries $\mathbf{U}_i$ we refer to Appendix \ref{app:reco_supersposition}.

\begin{theorem}
Suppose that the total number of reconstructed tracks is $O(n)$.
Then, Algorithm \ref{algo:reco_in_superposition} (quantum)  outputs the tracks reconstructed by the full CTF algorithm (seeding, track building, cleaning, and selection) with bounded-error probability in $\tilde{O}\left(n^{(c + 3) / 2} \right)$ time.
\label{thm:ctf}
\end{theorem}

\renewcommand{\arraystretch}{1.5}
\renewcommand{\tabcolsep}{0.2cm}

\begin{table*}[t]
\resizebox{\textwidth}{!}{%
\begin{tabular}{ccccc}
\hline\hline
\textbf{Tracking stages}                                                                                & \textbf{Input size}    & \textbf{Output size}& \textbf{Classical complexity}                                                                  & \textbf{Quantum complexity}                                                                               \\ \hline\hline
\textbf{Seeding}                                                                                        & $O(n)$                    & $k_{\text{seed}}$& \begin{tabular}[c]{@{}c@{}}$O\left(n^c\right)$\\ (Theorem \ref{thm:classicalseeding})\end{tabular}                       & \begin{tabular}[c]{@{}c@{}}$\tilde{O}\left(\sqrt{k_{\text{seed}} \cdot n^c }\right)$\\ (Theorem \ref{thm:quantumseeding})\end{tabular}      \\ \hline
\textbf{Track Building}                                                                                 & $k_{\text{seed}}+ O(n)$ & $k_{\text{cand}}$& \begin{tabular}[c]{@{}c@{}}$O(k_{\text{seed}} \cdot n)$\\ (Theorem \ref{thm:classicalbuilding})\end{tabular}  & \begin{tabular}[c]{@{}c@{}}$\tilde{O}\left(k_{\text{seed}} \cdot \sqrt{n}\right)$\\ (Theorem \ref{thm:quantumbuilding})\end{tabular}        \\ \hline
\textbf{Cleaning (original)}                                            &  $k_{\text{cand}}$ & $O(k_{\text{cand}})$& \begin{tabular}[c]{@{}c@{}}$O(k_{\text{cand}}^2)$ \\ (Theorem \ref{thm:classicalcleaning})\end{tabular}       & --                                                                                             \\ \hline
\textbf{Cleaning (improved)}                                            &  $k_{\text{cand}}$ & $O(k_{\text{cand}})$& \begin{tabular}[c]{@{}c@{}}$\tilde{O}(k_{\text{cand}})$ \\ (Theorem \ref{thm:newcleaning})\end{tabular} & --                                                                                             \\ \hline
\textbf{Selection}                                                                                      &  $O(k_{\text{cand}})$ & $O(k_{\text{cand}})$& \begin{tabular}[c]{@{}c@{}}$O(k_{\text{cand}})$\\ (Theorem \ref{thm:selection})\end{tabular}          & --                                                                                             \\ \hline\hline
\textbf{Full Reconstruction}                                                                            & $n$                    & $O(n^c)$& \begin{tabular}[c]{@{}c@{}}$O\left(n^{c + 1}\right)$\\ (Theorems \ref{thm:classicalseeding}, \ref{thm:classicalbuilding}, \ref{thm:newcleaning}, \ref{thm:selection})\end{tabular}  & \begin{tabular}[c]{@{}c@{}}$\tilde{O}\left(n^{c + 0.5}\right)$\\ (Theorems \ref{thm:quantumseeding}, \ref{thm:quantumbuilding}, \ref{thm:newcleaning}, \ref{thm:selection})\end{tabular}               \\ \hline
\textbf{\begin{tabular}[c]{@{}c@{}}Full Reconstruction with\\  $O(n)$ reconstructed tracks\end{tabular}} & $n$                    & $O(n)$& \begin{tabular}[c]{@{}c@{}}$O\left(n^{c + 1}\right)$\\ (Theorems \ref{thm:classicalseeding}, \ref{thm:classicalbuilding}, \ref{thm:newcleaning}, \ref{thm:selection})\end{tabular}  & \begin{tabular}[c]{@{}c@{}}$\tilde{O}\left(n^{(c + 3)/2}\right)$\\ (Theorem \ref{thm:ctf})\end{tabular}                       \\ \hline\hline
\end{tabular}%
}
\caption{
Summary of the results.
We present the complexity of the algorithms for each of the track reconstruction stages, both the classical and quantum versions.
$n$ is the number of charged particles present in the event record, $c$ is the number of hits used to form the seeds, $k_{\text{seed}}$ is the number of seeds generated, and $k_{\text{cand}}$ is the number of built candidate tracks.
The two rows for the track cleaning stage refer to the original version of \cite{JINST} and to the one we propose.
On the quantum side some entries are marked as ``--'' where we did not propose/expect a quantum algorithm with advantage over the classical one.
In the penultimate row we write the complexity of the full track reconstruction, assuming the four stages are executed sequentially.
We combine Theorems \ref{thm:classicalseeding}--\ref{thm:selection} using $k_{\text{seed}} = O(n^c)$ and $k_{\text{cand}}= O(k_{\text{seed}})$.
The final row shows that the quantum advantage can be further improved provided that the number of reconstructed tracks is $O(n)$.
}
\label{tab:summary}
\end{table*}

\section{Conclusions}

We have identified four fundamental routines present in local track reconstruction methods (seeding, track building, cleaning, and selection), and analyzed how each scales with the number of recorded hits, $n$, proposing quantum algorithms where we could find advantage (seeding and track building). 
The seeding stage, which runs on $O(n^c)$ time classically (Theorem \ref{thm:classicalseeding}), has a quantum computational complexity of $\tilde{O}\left(\sqrt{k_{\text{seed}} \cdot n^c } \right)$ (Theorem \ref{thm:quantumseeding}), where $c$ is the number of hits per seed and $k_{\text{seed}}$ is the total number of generated seeds.
Classical track building has complexity $O(k_{\text{seed}} \cdot n)$ (Theorem \ref{thm:classicalbuilding}), whereas we develop a quantum algorithm that scales as $\tilde{O}(k_{\text{seed}} \cdot \sqrt{n})$ (Theorem \ref{thm:quantumbuilding}).
These speedups are based on quantum search routines.
The Combinatorial Track Finder's version of the cleaning routine has complexity $O(k_{\text{cand}}^2)$ (Theorem \ref{thm:classicalcleaning}), where $k_{\text{cand}}$ is the number of processed track candidates, and we show that this can be improved to $\tilde{O}(k_{\text{cand}})$ (Theorem \ref{thm:newcleaning}) via a structured search scheme.
The selection stage was already optimal from the complexity perspective (Theorem \ref{thm:selection}).
If the four stages are run sequentially, the track building routine dominates the complexity of the reconstruction: $O(n^{c+1})$ classically and $\tilde{O}(n^{c + 0.5})$ quantumly.
However, we show that, if the number of reconstructed tracks is $O(n)$, we can combine all previous algorithms to perform track reconstruction in $\tilde{O}\left(n^{(c + 3) / 2} \right)$ time (Theorem \ref{thm:ctf}).
These results are summarized in Table \ref{tab:summary}.
We recall that all of our quantum algorithms assume access to a QRAM storing the classical hit data.

Our work develops a rigorous computational framework to analyze the track reconstruction problem that is both sufficiently abstract to encompass different experiments (like CMS, ATLAS, or Belle II) and sufficiently strong to predict the scaling of the Combinatorial Track Finder algorithm for high luminosity regimes.
Although the reached quantum speedups are evidently mild, we conjecture that they are the best possible while constrained to matching the exact output of the corresponding classical algorithm (up to bounded-error probability) at arbitrarily fine scales.
In other words, the direct quantization of (local) tracking methods may not be the best path to establish a significant advantage in quantum computing for HEP problems.
Instead, one may find more success by breaking the direct correspondence with the classical setting and designing completely new tracking algorithms that inherently take advantage of the features of quantum processors.
There have been proposals in this direction \cite{TrackAnnealers, TrackAnnealingInspired, Tuysuz2021} , but they have not yet shown clear evidence of quantum speedup (refer to Appendix \ref{app:comparison} for details).

In summary, we offer  the first rigorous quantum speedup for relevant HEP data processing tasks.
Moreover, our comprehensive analysis of the Combinatorial Track Finder algorithm reveals that classical improvements to the computational complexity are also possible.
And, even though asymptotic results may be of limited use for practical problems, and quantum hardware may still be far from being able to address big data problems, we hope our original approach to tracking can motivate further investigations on the potential of quantum computation to tackle the increasingly challenging, and potentially intractable classically,  High-Energy Physics data analysis problems.

\newpage

\begin{acknowledgments}
The authors would like to thank Felice Pantaleo for precious discussions about the classical Combinatorial Track Finder algorithm. 
Furthermore, the authors acknowledge  project \textit{QuantHEP – Quantum Computing Solutions for High-Energy Physics}, supported by the EU H2020 QuantERA ERA-NET Cofund in Quantum Technologies, and FCT -- Funda\c{c}\~{a}o para a Ci\^{e}ncia e a Tecnologia (QuantERA/0001/2019). DM, AK, SP, GQ, PB, JS, YO thank the support from FCT, namely through project UIDB/50008/2020. DM acknowledges the support from FCT through scholarship 2020.04677.BD.  MK thanks MikroTik for the scholarship administrated by the UL Foundation. AG has been partially supported by National Science Center under grant agreement 2019/32/T/ST6/00158 and 2019/33/B/ST6/02011. SP thanks the support from the la Caixa foundation through scholarship LCF/BQ/DR20/11790030. GQ thanks the support from FCT through project CEECIND/02474/2018.
\end{acknowledgments}

\appendix

\section{Other quantum algorithms for track reconstruction \label{app:comparison}}

Quantum-based algorithms have been previously developed for track reconstruction and related problems, but none of them show clear evidence of quantum speedup.
\cite{TrackAnnealers, TrackAnnealingInspired} approach tracking as a combinatorial optimization problem to be solved with quantum annealing, and \cite{Tuysuz2021} suggest a hybrid graph neural network model.
References \cite{TrackClusteringAnnealers} and \cite{Quiroz2021} address the related (but not equivalent) problems of track clustering and track recognition, respectively.
One should be cautious before drawing a direct comparison with the currently used tracking algorithms, as is done in this paper.
First, these studies do not guarantee the same or better output quality than the current classical approach.
Second, previous papers do not present a rigorous computational complexity analysis, thus it is unclear how they scale with the number of hits $n$.

In \cite{TrackAnnealers} and \cite{TrackAnnealingInspired} track reconstruction is formulated as a quadratic unconstrained binary optimization (QUBO) problem, which can be naturally mapped to a quantum annealer.
In \cite{TrackAnnealingInspired} the binary variables represent hit doublets $S_{ij}$, and in \cite{TrackAnnealers} hit triplets $S_{ijk}$.
With both approaches the idea is that in the optimal solution the variables assigned with $+1$ correspond to connections between hits left by the same particle.
The computational complexity of the QUBO polynomial formulation is dominated by the calculation of the ``reward'' coefficients (see original references for details),
\begin{gather}
    \sum_{ijk} b_{ijk} S_{ij}S_{jk} \text{ (doublets),}
    \\
    \sum_{ijklp} b_{ijklp} S_{ijk} S_{klp} \text{ (triplets)},
\end{gather}
where the coupling coefficients ($b_{ijk}$ or $b_{ijklp}$) are determined classically.
We see that the amount of coupling coefficients that need to be pre-processed to define the QUBO polynomial representing the full event is $\Theta(n^3)$ for the case of hit doublets 
and $\Theta(n^5)$ for the case of hit triplets.
It is unclear how the annealing time scales with $n$ and the required quality of the solution. The annealing time required to find the global minimum of the QUBO polynomial is likely to scale exponentially with respect to the number of variables in the polynomial \cite[Section~2.5]{TrackAnnealingInspired}.
\cite{TrackAnnealingInspired} consider dividing the problem into smaller sub-problems to enable embedding on current quantum annealers and reduce run-time, but the asymptotic scaling of these sub-problems is not analysed.

\cite{Tuysuz2021} proposed a hybrid classical-quantum graph neural network model.
After classical pre-processing, they get a graph of connections between the hits
of size $O(n^2)$.
Their model predicts the probability of each such connection linking to consecutive hits from the 
same particle.
The complexity of this algorithm will also depend on the number of hidden classical $n_C$ and quantum $n_Q$ dimensions and on the number of training iterations $n_I$.
Although they show that small values of $n_C$, $n_Q$, and $n_I$ are enough to get relatively accurate results for their (toy) datasets, it is difficult to assess whether the proposed algorithm would achieve similar results to the CTF algorithm, and it is unclear whether increasing the size of the network ($n_C$, $n_Q$, $n_I$) would improve the results sufficiently.

\section{More on the tracking problem \label{app:trackingmodel}}

The work presented in the main text is based on a simplified model of the problem of track reconstruction.
Here we describe in more detail the assumptions behind that model.

\textit{Assumptions about the particles' trajectories.}
Let $n$ denote the number of charged particles present in the event record.
Their trajectories originate from a fixed interaction region, but we do not assume that they come from a single collision spot.
The detector is immersed in a quasi-uniform co-axial magnetic field, so we expect the particles to follow helical trajectories aligned with the field's direction.
These trajectories can be described by five parameters \cite{Miao07}.
Let $\mathcal{P} \subset \mathbb{R}^5$ be the five-dimensional cuboid corresponding to the trajectories' parameter space.
Each experiment that produces a sequence of particle trajectories is governed by a physical process that implicitly selects these parameters for each trajectory; in the forthcoming, we call this procedure an \emph{event}.
We model this formally as a random variable $\pi : \Omega \to \mathcal{P}$ that selects a parameter with respect to a probability space $(\Omega, \mathcal{F}, P)$ that accounts for various physical parameters, such as noise, etc..
We make the mild assumption that $\pi$ follows a probability distribution $p_{\pi}$ on $\mathcal{P}$ that is strictly positive.
With our model, an event is generated by drawing $n$ random samples $\pi_1, \ldots, \pi_n$, each following $p_{\pi}$ (that is, we treat them as i.i.d. random variables).

\textit{Assumptions about the detector's layers.}
The detector has a fixed geometry with a discrete set of sensor layers.
We consider that the detector has $L$ cylindrical layers aligned with the beam line.
The layers are indexed from $0$ to $L-1$ (the most inner layers having the smallest indices).
We assume that each particle traverses every layer and that they never return to a previously visited layer. 
We assume the layers to be continuous two-dimensional surfaces $\mathcal{C}_l$, $l \in \{0, \ldots, L-1\}$.

\textit{Assumptions about the hits' data.}
As a particle traverses a layer, it leaves a complex signal resulting from the interaction with the sensor's pixels.
We treat these signals simply as three-dimensional points, called hits.
We assume that the granularity of the sensors is high enough such that each detected hit can be differentiated from other hits. 
Since each trajectory leaves a unique hit on each layer $l$, formally we have a continuous map $H_l : \mathcal{P} \to \mathcal{C}_l$, relating each trajectory to its point of intersection with layer $l$.
At every layer, we identify the hits by labels from $\{0, \ldots, n-1\}$.
We use the notation $\mathbf{m}_{l,j}$ for the coordinates of $j$th hit in layer $l$.
It is possible that some hits are not measured at all due to sensor inefficiencies.
That is, we do not necessarily have a hit $\mathbf{m}_{l,j}$ for every pair $(l,j)$. This means that we may not be able to tell the exact value of $n$ directly from the event record, as it is possible that there is no layer registering all particles. In that case, we would be indexing the hits with labels from $\{0, \ldots, n^*-1\}$, where $n^*$ is the largest number of hits measured in any layer. We consider that $n^*=n$ for simplicity, but every result in this paper would hold the same as long as $n^*=\Theta(n)$.

Under these assumptions, we expect the number of measured hits per layer to grow proportionally to $n$.
Moreover, we can establish the following useful lemma:

\begin{proofT}{Lemma \ref{thm:almostallevents}}
The trajectories are determined by the parameters in $\mathcal{P}$.
If $S \subset \mathcal{C}_l$ is open and non-empty, then $U_S := H_l^{-1}(S) \in \mathcal{P}$ is open and non-empty and by assumption, the probability that a trajectory has parameters in $U_S$ is $p_S := p_{\pi}(U_S) > 0$.
If $\pi_1, \ldots, \pi_n$ are $n$ random samples of parameters drawn without replacement and following $p_{\pi}$, then the strong law of large numbers implies that almost surely, the number of parameters sampled from $U_S$ grows as $\Theta(n)$.
Applying $H_l$ to this gives the growth of the number of hits in $S$.
\end{proofT}

In our model we have omitted some details that, despite being a crucial part of real particle physics experiments, do not significantly influence our complexity analysis.
We now comment on how some of these assumptions could be relaxed in the context of our analysis. 

First, we should note that most local tracking algorithms, including the CTF, reconstruct tracks by multiple iterations.
That is, the sequence of the four computational routines (seeding, track building, cleaning, and selection) is called several times for the same event record.
The idea of this iterative tracking is that the initial iterations search for the tracks that are easiest to find (high transverse momentum, and produced near the interaction region).
After each iteration, the hits associated to the reconstructed tracks are removed, thereby simplifying the subsequent iterations. 
As far as our complexity analysis is concerned, the most significant modification from iteration to iteration is the number of hits used to form seeds.
In \cite{JINST}, the first iteration forms seeds with hit triplets.
But in some subsequent iterations seeds are formed by picking just two hits, as we can use the results of the previous iteration to reconstruct the collision vertices, which serve as the ``third hit''.
In our work, we have analysed the general case of forming seeds with $c$ hits.

Regarding our model for the generation of the trajectories, we assumed that the probability distribution $p_{\pi}$ on parameter space $\mathcal{P}$ was strictly positive.
This was not deduced from explicit particle physics calculations, but is a rather lax assumption that includes the seemingly reasonable assumption that no scattering direction is forbidden.
We've also assumed that an event is generated by drawing $n$ random samples $\pi_1, \ldots \pi_n$, each following $p_{\pi}$.
Underlying this there is the physical assumption that the trajectories of the charged particles are treated as uncorrelated.

We also assumed that each hit can always be differentiated from other hits.
In the $n \rightarrow \infty$ limit, this requires infinite detector's granularity.
But the underlying idea is simply that the experimental precision keeps up with the increase in the beam's luminosity in the sense that we can always distinguish one detection from another.
Indeed, if it were not the case there would be no point in going to very high luminosities.

Another of our assumptions was that the layers were continuous surfaces, as this made the description of the algorithm clearer (especially for the track building stage).
In reality, the layers are formed by overlapping sensor modules.
This means that it is possible for a particle to leave more than one detection per layer.
To accommodate this possibility, at each layer the CTF selects compatible modules, which are the ones whose boundaries are up to a given distance from the predicted measurement.
These modules are divided into module groups in such a way that no two modules in the same group overlap.
Only the best measurement from each group is considered to integrate the track candidate.
Again, we only allow up to five new track candidates per step. 
At each iteration there are $O(1)$ module groups, each with $O(n)$ hits, so the claimed complexity remains the same.

Furthermore, we assumed that the detector was a collection of $L = O(1)$ cylindrical layers.
This is simplified description of real detectors.
For example,
the CMS tracker has a barrel-like shape, with thirteen cylindrical layers aligned with the beam
line and fourteen disk layers in the transverse plane. We did not include the detailed geometry
of the detector in our discussion in order to simplify the exposition. In fact, our analysis holds
for any disposition of layers as long as we assume that each particle only traverses one layer
at a time and that their trajectories do not return to a previously visited layer.

We have considered that the hit data is given in the form of three-dimensional points.
Actually, as a charged particle traverses a layer, it activates multiple sensor pixels.
Then, the signals in neighbouring pixels are grouped together to form three-dimensional clusters.
The centroid of each cluster determines a hit's position.
But the cluster shape also carries information.
In particular, in some cases it is possible to exclude a hit from a given track based on the incompatibility between the hit's cluster shape and the track's trajectory.
We may see this as a motivation to think about the case where $k_{\text{seed}}=O(n^a)$ with $a<c$ -- even though the cluster shape information does not provide a mean to find the good seeds faster, it guarantees that we can recognize them.

In summary, we see that several of our simplifying assumptions could be lifted without changing our conclusions. 
Arguably, the strongest assumption was ignoring the hits' cluster shape information, as that might be used to exclude $k_{\text{seed}}=O(n^c)$ as a worst-case scenario.

\section{The seeding algorithms \label{app:seedingalgorithm}}

\subsection{Classical}

The purpose of the seeding stage is to provide initial track candidates, formed by $c$ hits, and their trajectory parameters.
The CTF algorithm \cite{JINST} generates seeds by selecting $c$-tuplets of hits from the $c$ most inner layers.

We now describe this process in detail.
First, we search over the first two seeding layers of the detector for pairs of hits that are compatible with the seeding criteria.
These criteria include:
\begin{enumerate}
    \item \textit{Minimum transverse momentum.} 
    We impose a minimum value, $p_0$, on the transverse component of the particle's momentum with respect to the direction of the magnetic field, $p_T$.
    Recall that in a uniform magnetic field the trajectory of a particle is an helix aligned with the field and $p_T$ is proportional to the radius of the helix.
    So, $p_T > p_0$ is equivalent to geometric condition that we only accept trajectories with a minimum radius.
    \item \textit{Maximum transverse and longitudinal distance of closest approach to the beam-spot.} 
    The beam-spot, $\mathbf{r}_0$, is the point where we expect that the collisions takes place (this is estimated independently of track reconstruction -- see \cite[Section~6]{JINST}). 
    It does not mean that all trajectories originate from that exact point, but we disregard trajectories that are far away from it.
    More precisely, we enforce that all trajectories cross a cylinder centered at $\mathbf{r}_0$ aligned with the beam axis with radius $\rho_0$ and height $z_0$ (these are the maximum transverse and longitudinal distances, respectively).
\end{enumerate}
For the purposes of our complexity analysis, the most important point about these criteria is that they can be checked for any $c$-tuplet in $O(c)=O(1)$ time.
For each pair of hits $(\mathbf{m}_{0,j_0},\mathbf{m}_{1,j_1})$ that we formed, we search over the third layer for hits compatible with the seeding constraints.
That is, we select a hit $\mathbf{m}_{2,j_2}$ if the trajectory defined by $(\mathbf{m}_{0,j_0},\mathbf{m}_{1,j_1}, \mathbf{m}_{2,j_2})$ satisfies the conditions of (1) and (2) above.
Then, for each hit triplet $(\mathbf{m}_{0,j_0},\mathbf{m}_{1,j_1}, \mathbf{m}_{2,j_2})$, we search over the fourth layer for hits $\mathbf{m}_{3,j_3}$ such that the trajectory that best fits $(\mathbf{m}_{0,j_0},\mathbf{m}_{1,j_1}, \mathbf{m}_{2,j_2}, \mathbf{m}_{3,j_3})$ satisfies the seeding constraints.
This process is repeated until the $c$th layer is reached.
For a for a pseudocode representation see Algorithm \ref{algo:seeding} in Appendix \ref{app:pseudocodes}.

\begin{proofT}{Theorem \ref{thm:classicalseeding}}
Suppose that we have built a seed up to layer $l \in {1, \ldots c-2}$.
Selecting the hits in the $(l+1)$th layer that constitute valid continuations for the seed takes $O(n)$ time: there are $O(n)$ candidate hits and, for each $(l+1)$-tuplet, verifying if it satisfies the seeding criteria takes $O(1)$ time.
So, if there are $N_l$ seeds at layer $l$, iterating over to layer $l+1$ takes $O(N_l \cdot n)$ time.

We now estimate $N_l$.
For each seed that we have built at layer $l$, the seeding constraints define a (non-empty) region over the $(l+1)-$th layer where any hit could be used to continue that seed -- see Figure 3 of the main text.
According to Lemma \ref{thm:almostallevents}, in that region we expect there to be $\Theta(n)$ hits.
So, the number of formed seeds will increase by a multiplicative factor of $\Theta(n)$ when going from layer $l$ to layer $l+1$.
As we start with $n$ seeds at layer $0$, we end up with $N_l = O(n^{l+1})$.

So, the complexity of the algorithm is
\begin{equation}
    n + \sum_{l = 1}^{c-1} N_{l-1} \cdot  O\left( n \right) = O(n^c).
\end{equation}

\end{proofT}

We point out that during the proof of Theorem \ref{thm:classicalseeding} we have also established that Algorithm \ref{algo:seeding} forms at most $\Theta(n^c)$ seeds.

We recall that the version of the CTF presented in \cite{JINST} forms seeds with just three hits.
Theorem \ref{thm:classicalseeding} applies to that case by making $c=3$.
We are trying to fit the tracks into helices aligned with the detector axis.
So, we could not fix the five parameters defining such helices with fewer than three points without critically relying on the estimate for the beam-spot.
Actually, in general three points determine a countable family of such helices. 
If we assume that the trajectories do not realize ``a full turn'' between these points, this degeneracy is broken. However, then we do not have the guarantee that there is an helix passing exactly through the three points. 
In practice, as there are experimental uncertainties about the hits' positions, this is not a concern.

The seeding strategy that we have described was used at tracking programme of the CMS collaboration during the 2016 LHC  run \cite{Sguazzoni2016}.
Meanwhile, CMS has adopted a different seeding method \cite{Bocci2020}, based on the concept of cellular automata.
This more sophisticated algorithm yields seeds of various sizes, and already includes a cleaning and selection phases.
Perhaps more importantly, its design is extremely parallelizable.
This means that it can become very efficient in terms of speed, scaling better than the corresponding computational complexity.

\subsection{Quantum}

Consider a unitary transformation $\mathbf{U}_{\text{seed}}$ that recognizes if a hit $c$-tuplet forms a valid seed.
That is, given a state $\ket{ \mathbf{m}_{0,j_0}, \ldots, \mathbf{m}_{c-1,j_{c-1}} }$, $\mathbf{U}_{\text{seed}}$ applies a $-1$ phase if $(\mathbf{m}_{0,j_0}, \ldots, \mathbf{m}_{c-1,j_{c-1}})$ passes the seeding stage, and otherwise the state is left unchanged.
Since any classical computation can be simulated by a quantum computer \cite{NielsenChuang}, this is clearly possible.
Moreover, because we can recognize if a hit $c$-tuplet constitutes a valid seed with a $O(1)$-sized circuit, we can also build a quantum circuit for $\mathbf{U}_{\text{seed}}$ using $O(1)$ gates.
Using $\mathbf{U}_{\text{seed}}$ and $\mathbf{Q}$, it is straightforward to form a unitary transformation $\mathbf{O}_{\text{seed}}$  acting on $\{\ket{0}, \ket{1}\}^{\otimes 3 \log n}$ (possibly along some ancillary qubits)
that marks the state $\ket{j_0, \ldots, j_{c-1}}$ with a $-1$ phase if the corresponding hit $c$-tuplet constitutes a good seed.
If any of the pairs of indices $(0, j_0), \ldots, (c, j_{c-1})$ does not correspond to a hit, we assume that $\mathbf{O}_{\text{seed}}$ leaves the state $\ket{ j_0, \ldots, j_{c-1} }$ unchanged.
We can run the circuit for $\mathbf{O}_{\text{seed}}$  in $O(\log(n))$ time.

We start by preparing all $c$-tuplets in superposition.
For simplicity, we assume that $n$ is a power of two.
Starting from the all-zero state, we can do this by applying $c \log n$ parallel  single-qubit Hadamard gates (also known as the Walsh-Hadamard transform, $\mathbf{H}$).
Now define $\theta$ and $m$ as
\begin{equation}
    \theta = \arcsin \left( \sqrt{ \frac{k_{\text{seed}}}{n^c} } \right) , \quad
    m = \left\lfloor \frac{\pi}{4 \theta} \right\rfloor.
\label{eq:rdefinition}
\end{equation}
From Grover's algorithm, 
\begin{theorem}[Quantum search \cite{Grover1996, BoyerSearch}]
Let $m$ and $\mathbf{G}$ be defined as in \eqref{eq:rdefinition} and \eqref{eq:Gdefinition}, respectively.
Then, if we measure the state
\begin{equation}
    \mathbf{G}^m \cdot \left( \frac{1}{\sqrt{n^c}} \sum_{j_0, \ldots, j_{c-1} = 0}^{n-1}  \ket{ j_0, \ldots, j_{c-1} } \right)
    \label{eq:groversalgorithm}
\end{equation}
in the computational basis we will find a good seed (i.e., a $c$-tuplet $(j_0, \ldots, j_{c-1})$ such that $(\mathbf{m}_{0,j_0}, \ldots,$ $\mathbf{m}_{c-1,j_{c-1}})$ passes the seeding stage) with probability at least $1/2$.
\label{thm:grover}
\end{theorem}

Preparing the state \eqref{eq:groversalgorithm} involves calling the operator $\mathbf{G}$ $m=O(\sqrt{n^c/k_{\text{seed}}})$ times, representing a complexity advantage over what we could do classically.
Note, however, that applying Grover's algorithm requires determining $m$, which we cannot do since we do not know \emph{a priori} what is the value of $k_{\text{seed}}$.
For that purpose, we can use the quantum counting algorithm of Brassard, H{\o}yer, and Tapp \cite{quantumcounting}:

\begin{theorem}[Quantum counting \cite{quantumcounting}]
There is a quantum algorithm that outputs $k_{\mathrm{seed}}$ with probability at least $3/4$, using an expected number of $O(\sqrt{n^c k_{\mathrm{seed}}})$ calls to $\mathbf{G}$.
\label{thm:quantumcounting}
\end{theorem}

\noindent
Combining these techniques, we propose performing seeding with Algorithm \ref{algo:quantumseeding}.

\begin{proofT}{Theorem \ref{thm:quantumseeding}}
Assume that quantum counting succeeds, that is, we have correctly estimated $\tilde{k}=k_{\text{seed}}$ in step \ref{step:counting} in $\tilde{O}(\sqrt{n^c k_{\text{seed}}})$ time.
The probability of sampling a \emph{new} good seed in step \ref{step:groverstate} after having already found $k$ of them is (Theorem \ref{thm:grover})
\begin{equation}
    \frac{1}{2} \frac{k_{\text{seed}} - k}{k_{\text{seed}}}.
\end{equation}
Then, determining the expected time to find all good seeds is equivalent to the coupon collector's problem.
In particular, the probability that we run step \ref{step:groverstate} more than $10 k_{\text{seed}} \log k_{\text{seed}}$ times is less than $1/4$.
That is, with probability at least $3/4$ we spend $\tilde{O}(\sqrt{n^c \cdot k_{\text{seed}}})$ time on loop \ref{step:seedwhile}-\ref{step:seedoutput}.
Since quantum counting succeeds with probability at least $3/4$ (Theorem \ref{thm:quantumcounting}), Algorithm \ref{algo:quantumseeding} outputs all seeds in $\tilde{O}(\sqrt{n^c \cdot k_{\text{seed}}})$ time with probability no less than $1/2$.
\end{proofT}

\noindent
If $k_{\text{seed}} = O(n^a)$, then we can perform seed generation up to bounded error with complexity
\begin{equation}
    \tilde{O} \left( n^{\frac{c + a}{2}} \right).
\end{equation}
In the worst-case scenario, this shows no advantage over the classical algorithm (as is expected from Lemma \ref{thm:almostallevents}).
But for any $a<c$ we reach a lower complexity than the classical seeding (Theorem \ref{thm:classicalseeding}).

\section{The track building algorithms \label{app:trackbuilding}}

\subsection{Classical \label{sec:trackbuildingclassical}}

The track building stage extrapolates the seeds’ trajectories along the expected path of the particle and builds track candidates by adding compatible hits from successive detector layers, updating the parameters at each layer.
More precisely, the track building strategy is based on the combinatorial Kalman filter \cite{KalmanFilterFruhwirth}, which in turn is an adaptation of the Kalman filter \cite{Kalman} for tracking problems.
We now describe this method.

Say that a trajectory at layer $l-1$ is described by a five vector $\mathbf{p}_{l-1}$.
The propagated state vector $\mathbf{p}_l$ at next layer is modelled by the system equation
\begin{equation}
\mathbf{p}_l =  \mathbf{F}_l \mathbf{p}_{l-1} + \mathbf{w}_l.
\label{eq:process}
\end{equation}
$\mathbf{F}_l$, known as the process matrix, describes the propagation of a charged particle in a uniform magnetic field from layer $l-1$ to $l$.
$\mathbf{w}_l$ is a random variable called process noise.
A measurement $\mathbf{m}_l$ at layer $l$ is given by
\begin{equation}
\mathbf{m}_l =  \mathbf{H}_l \mathbf{p}_{l-1} + \mathbf{e}_l,
\label{eq:measurement} 
\end{equation}
where $\mathbf{H}_l$ is the measurement matrix and $\mathbf{e}_l$ is the measurement noise.
We assume that we know the covariance matrices for the process and measurement noises.
Note that, in general, we could replace equations \eqref{eq:process} and \eqref{eq:measurement} by non-linear relations.
But the linear model usually suits the purpose of track reconstruction.

Suppose that we have built a track up to layer $l-1$ with the measurements (i.e., hits) $\mathbf{m}_{0, j_0}, \mathbf{m}_{1, j_1}, \ldots,$ $\mathbf{m}_{l-1,j_{l-1}}$.
With this information, we describe our prediction of the trajectory at this layer by a state vector $\mathbf{p}_{l-1|l-1}$ and corresponding covariance matrix.
Without knowing which hit from layer $l$ belongs to this track, we predict that the state vector at layer $l$ is
\begin{equation}
\mathbf{p}_{l|l-1} = \mathbf{F}_l \mathbf{p}_{l-1|l-1}.
\label{eq:predictedstate}
\end{equation} 
We say that the predicted measurement at layer $l$ is
\begin{equation}
\mathbf{m}_{l|l-1} = \mathbf{H}_l \mathbf{p}_{l|l-1}.
\label{eq:predictedmeasurement}
\end{equation}
This would be the location of the $l$th hit if we had perfect knowledge of the trajectory and there were no process/measurement errors.
In reality, we do not expect to find any hit exactly at this predicted measurement.
When considering an actual measurement $\mathbf{m}_{l,j}$, we say that the predicted residual is
\begin{equation}
\mathbf{r}_{l|l-1} \left( \mathbf{m}_{l,j} \right) = \mathbf{m}_{l,j} - \mathbf{m}_{l|l-1}
\end{equation}
The predicted $\chi^2$ value is defined as
\begin{equation}
\chi^2_{l|l-1} \left( \mathbf{m}_{l,j} \right) = \mathbf{r}_{l|l-1} \left( \mathbf{m}_{l,j} \right)^T \mathbf{R}_{l|l-1}^{-1} \mathbf{r}_{l|l-1} \left( \mathbf{m}_{l,j} \right),
\label{eq:predictedchi2}
\end{equation}
where $\mathbf{R}_{l|l-1}$ is the covariance matrix of the predicted residual.
Intuitively, a high $\chi^2$ value tells us that the measurement is unlikely to belong to the track.
Then, when evaluating which hit to add to the track, only the ones whose predicted $\chi^2$ value is below some fixed threshold $\chi_0^2$ pass to the filtering phase.
Suppose that $\mathbf{m}_{l,j}$ satisfies this criterion.
Based on this measurement, we update the state vector prediction to
\begin{equation}
\mathbf{p}_{l|l} = \mathbf{p}_{l|l-1} + \mathbf{K}_l \mathbf{r}_{l|l-1} \left( \mathbf{m}_{l,j} \right),
\end{equation}
where $\mathbf{K}_l$ is the Kalman gain matrix, which is calculated based on the covariance matrices of state vector, the process noise and the measurement noise (see \cite{KalmanFilterFruhwirth} for explicit expression).
We say that the filtered residual for this measurement is
\begin{equation}
\mathbf{r}_{l|l} \left( \mathbf{m}_{l,j} \right) = \mathbf{m}_{l,j} - \mathbf{H}_l \mathbf{p}_{l|l}.
\end{equation}
The filtered $\chi^2$ value is
\begin{equation}
\chi^2_{l|l} \left( \mathbf{m}_{l,j} \right) = \mathbf{r}_{l|l} \left( \mathbf{m}_{l,j} \right)^T \mathbf{R}_{l|l}^{-1} \mathbf{r}_{l|l} \left( \mathbf{m}_{l,j} \right),
\label{eq:filterchi2}
\end{equation}
$\mathbf{R}_{l|l}$ being the covariance matrix of the filtered residual.
One can show that the predicted and filtered $\chi^2$ values are actually identical (see \cite{KalmanFilterFruhwirth}), that is,
\begin{equation}
\chi^2_{l|l-1} \left( \mathbf{m}_l \right)
=
\chi^2_{l|l} \left( \mathbf{m}_l \right)
, \forall \mathbf{m}_l \in \mathbb{R}^3.
\end{equation}
This means that we can determine the filtered $\chi^2$ value without explicitly updating the trajectory.
The total $\chi^2$ value of the track at layer $l$ is the sum of the filtered (or predicted) $\chi^2$ values from all previously visited layers
\begin{equation}
\chi^2_{\leq l} 
(\mathbf{m}_{0, j_0},  \ldots, \mathbf{m}_{l,j_l})
= \sum_{i=0}^l \chi^2_{i|i} (\mathbf{m}_{i, j_i}).
\end{equation}

In general, we may have several hits passing to the filtering phase.
As we are not sure which one truly belongs to the track, we form new candidate tracks each including a different hit.
These tracks are then followed independently.
Also, to accommodate the possibility of detection inefficiencies the CTF permits adding a ``ghost hit'' if no suitable hit is found.
However, to avoid a rapid increase in the number of tracks, we impose a limit of $\lambda$ tracks retained at each step (the default in \cite{JINST} being $\lambda=5$).
If at any point this limit is surpassed, we abandon the worst tracks.
To decide this, each track candidate is attributed a quality score $q_l$ of the form
\begin{equation}
q_l = l - m_{ghost} -  \omega \cdot \chi^2_{\leq l},
\label{eq:qualityscore}
\end{equation}
where $m_{ghost}$ is the number of ghost hits included in the track and $\omega$ is some configurable weight (we omitted the dependence on the measurements).
At any step we can discard a candidate track if it contains too many ghost hits or the total $\chi^2$ value exceeds a given threshold.
Otherwise, the procedure is continued until the end of the detector is reached (that is, we arrive at $l=L-1$).
The quality score \eqref{eq:qualityscore} at that point is said to be the quality score of the track candidate. 
The tracks that reach this step are accepted for the next stage of the CTF algorithm.
The steps of the track building stage are summarized in Algorithm \ref{algo:finding}.

Before proving Theorem \ref{thm:classicalbuilding}, it is important to understand how many candidate hits pass to the filtering phase.
The space of points with acceptable predicted $\chi^2$ value (equation \eqref{eq:predictedchi2})
\begin{equation}
\{ \mathbf{m_l} \in \mathbb{R}^3 : \chi^2_{l|l-1} \left( \mathbf{m}_l \right) < \chi_0^2 \}
\end{equation}
is an ellipsoid around the predicted measurement.
The intersection of this ellipsoid with the layer's surface yields a region in that layer whose area is independent of $n$.
By Lemma \ref{thm:almostallevents}, we may find $\Theta(n)$ hits in that region.
Therefore,
\begin{lemma}
The filtering step takes $O(n)$ time.
\label{lem:filteringtime}
\end{lemma}

\begin{proofT}{Theorem \ref{thm:classicalbuilding}}
Starting from a single seed, we only propagate up to $\lambda=O(1)$ tracks from layer to layer.
For each of these, the analytical continuation of the trajectory from one layer to another (equations~\eqref{eq:predictedstate} and \eqref{eq:predictedmeasurement}) is performed in $O(1)$ time.
As we have seen with Lemma \ref{lem:filteringtime}, performing filtering requires $O(n)$ time per track candidate.
Finally, in $O(n)$ time we can determine the $\lambda$ tracks with best quality score \eqref{eq:qualityscore} that are propagated to the next layer.
The number of layers $L$ is $O(1)$.
Combining everything, we reach the find a complexity of $O(k_{\text{seed}} \cdot n)$.
\end{proofT}

\subsection{Quantum}

We have seen that, at each step of the track building stage, $O(n)$ hits undergo the filtering step, while only at most $\lambda=O(1)$ of them form new track candidates.
Our idea is to use quantum search to perform filtering, reducing its complexity from $O(n)$ (Lemma \ref{lem:filteringtime})to $\tilde{O}(\sqrt{n})$.

Suppose that we have followed a track up to layer $l-1$ according to the track building method described in Section \ref{sec:trackbuildingclassical}.
In particular, we have evaluated the predicted state vector and corresponding covariance matrix.
Based on this information, we can calculate predicted $\chi^2$ value (equation \eqref{eq:predictedchi2}) for any measurement in layer $l$ in $O(1)$ time.
Let $\mathbf{O}_{\chi}$ be a unitary transformation that, given the index of a measurement, calculates the predicted $\chi^2$ value of adding that measurement to the track:
\begin{equation} 
    \mathbf{O}_{\chi}
    \ket{l,j}  
    \ket{ x } 
    = 
    \ket{l,j}  
    \ket{x \oplus  \chi^2_{l|l-1}(\mathbf{m}_{l,j}) }.
\end{equation}
Like in the seeding algorithm, we can build a quantum circuit for $\mathbf{O}_{\chi}$ using the classical circuit to compute the predicted $\chi^2$ value and the QRAM operator $\mathbf{Q}$, requiring a total of $O(\log n)$ gates.
Using $\mathbf{O}_{\chi}$ we can build a quantum circuit $\mathbf{O}_{\text{find}}$ that marks a state $\ket{l,j}\ket{y}$ if $\chi^2_{l|l-1}(\mathbf{m}_{l,j})$ is smaller than the threshold $y$
\begin{equation}
\mathbf{O}_{\text{find}} \ket{l,j} \ket{y} 
=
\begin{cases}
- \ket{l,j} \ket{y}
, &\text{if } \chi^2_{l|l-1}(\mathbf{m}_{l,j}) < y \\
+ \ket{l,j} \ket{y}
, &\text{otherwise.} 
\end{cases}
\end{equation}
By the quantum minimum finding algorithm of D\"{u}rr and H{\o}yer \cite{DurrMinimum}, we can find the measurement $\mathbf{m}_{l,j}$ that minimizes $\chi^2_{l|l-1}$ with $O(\sqrt{n})$ calls to $\mathbf{O}_{find}$:

\begin{theorem}[Quantum minimum finding \cite{BoyerSearch}]
If there is a measurement $\mathbf{m}_{l,j}$ such that $\chi^2_{l|l-1}(\mathbf{m}_{l,j}) < \chi^2_0$, Algorithm \ref{algo:quantumminimumfinding} finds the measurement that minimizes $\chi^2_{l|l-1}$ with probability at least $1/2$ in $\tilde{O}(\sqrt{n})$ time.
\label{thm:quantumminimumfinding}
\end{theorem}

With this result, our strategy for track building becomes the following.
Starting from a single seed, we do track finding by propagating up to $\lambda$ tracks from layer to layer.
For each of these tracks, we apply quantum minimum finding $\lambda$ times to find the $\lambda$ measurements with lowest predicted $\chi^2$ value (after we have found a minimum of $\chi^2_{l|l-1}$ we can arbitrarily increase the $\chi^2$ value of that measurement to ensure that we do not find it again in the following run of quantum minimum finding).
Out of the up to $\lambda^2$ resulting track candidates, we select the $\lambda$ ones with best quality score and continue propagating those.
Note that this implies applying quantum minimum finding up to $L \lambda^2 n^c$ times, which means that the probability of correctly reproducing the result of the classical track building decreases with $n$.
Fortunately, we can make the probability of success bounded by always repeating the quantum minimum finding routine $O(\log n)$ times.
We propose doing track building as in Algorithm \ref{algo:quantumfinding}.

\begin{proofT}{Theorem \ref{thm:quantumbuilding}}
In Algorithm \ref{algo:quantumfinding}, instead of looping over the candidate measurements at each layer (line \ref{step:hitloop} in Algorithm \ref{algo:finding}), we find the best measurements with a quantum minimum finding routine.
We stop after having selected $\lambda$ measurements per candidate track as we know that only up to $\lambda$ tracks are kept at each layer (per seed).
Each run of quantum minimum finding takes $\tilde{O}(\sqrt{n})$ time -- Theorem \ref{thm:quantumminimumfinding}.
So, the result holds as long as we show that the probability of success is bounded by $1/2$.
The probability that we fail to select the best available measurement in steps \ref{step:qmflogtimes}-\ref{step:measurementselect} is upper bounded by
\begin{equation}
    \frac{1}{3 L \lambda^2 n^c}.
\end{equation}
Then, the probability that we do not fail any of the $L \lambda^2 n^3$ times we run steps \ref{step:qmflogtimes}-\ref{step:measurementselect} is lower bounded by
\begin{equation}
    \left(1 - \frac{1}{3 L \lambda^2 n^c}\right)^{L \lambda^2 n^c} 
    \geq
    \frac{2}{3}.
\end{equation}
\end{proofT}

\section{The cleaning algorithms \label{app:cleaning}}

\begin{figure*}[t]
\includegraphics[width=0.8\linewidth]{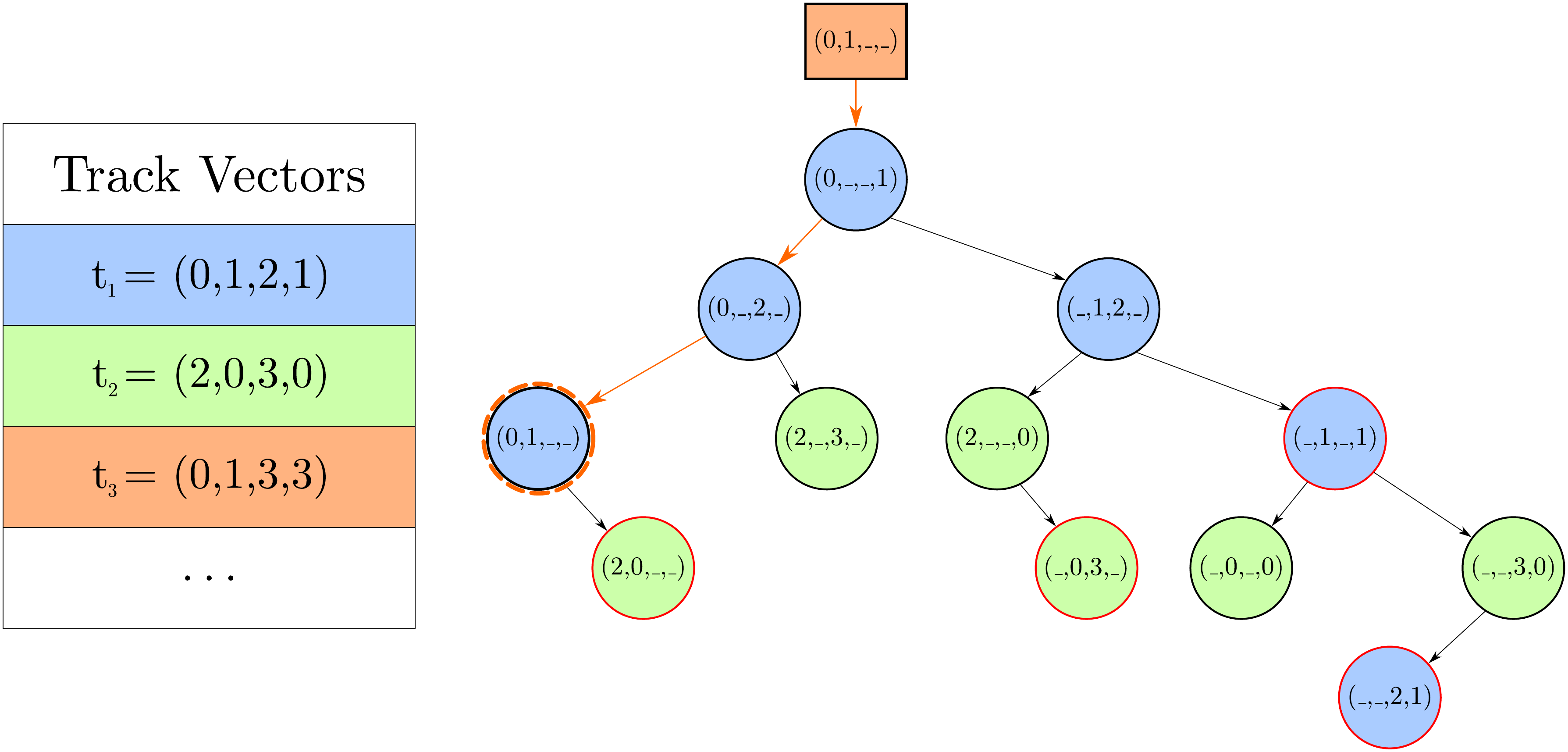}
\caption{Cleaning with $r$-tuples tree.
For this example, the first three elements of the sorted list of track vectors are $t_1 = (0,1,2,1)$, $t_2=(2,0,3,0)$, and $t_3=(0,1,3,3)$.
Suppose we want to exclude tracks that share two or more hits.
We have build a red-black tree with the $2$-tuples of $t_1$ and $t_2$ (blue and green circles, respectively).
The line of the circles is red or black according to the colour of the corresponding node (see \cite{IntroAlgs} for construction of red-black trees).
In this illustration, we are searching for $2$-tuples of $t_3$ in the tree.
We see that $2$-tuple $(0,1,\_,\_)$ is already present in the tree -- the path with orange leads to a node with that $2$-tuple.
So, $t_3$ is not going to be included in the output.
}
\label{fig:redblacktree}
\end{figure*}

CTF's cleaning algorithm compares every pair of tracks coming from the finding stage.
This approach does not take into account any structure of the tracks.
Indeed, it would work the same if instead of calculating the fraction of shared hits we were calling a black box that outputted ``clean/not clean'' when given two tracks.
We now present a different way to perform cleaning that takes better advantage of the structure of the problem.

We begin by reviewing the case where all the candidate tracks have $L$ hits, that is, each candidate track contains exactly one hit per layer.
Then, each candidate track can be uniquely identified with a vector in $\{0, \ldots, n-1\}^L$.
As an example, if $L=4$ and a given track $t$ contains the zeroth hit from the first layer, the second hit from the second layer, the fourth hit from the third layer, and the fourth hit from the fourth layer, its corresponding \emph{track vector} is $(0,2,4,4)$.
With $f$ being the maximum allowed fraction of shared hits, define $r= \lceil f L \rceil$.
We say a vector of length $L$ is an $r$-tuple of a track if it is equal to the track vector at $r$ entries and contains the symbol ``$\_$'' at the others.
For example, $(0,2,\_,\_)$ and $(0,\_,4,\_)$ are $2$-tuples of the track $t$ mentioned above.
Note that there are $\binom{L}{r}=O(1)$ such $r$-tuples.
Two tracks exceed the allowed fraction of shared hits if and only if they have (at least) $r$ hits in common, that is, if they have a matching $r$-tuple.

The algorithm starts by sorting the candidate tracks by quality score.
This way, if we need to discard one of two tracks we choose the one that is further down the list.
We then iterate over the sorted tracks.
Evidently, the first track $t_1$ is going to be included in the output.
We create a self-balancing binary search tree $\mathcal{T}$ (like a red-black tree -- see, for example, \cite{IntroAlgs}) containing all of the $r$-tuples of $t_1$ (with some induced order on the $r$-tuples).
We then move to the second track in the list $t_2$.
For every $r$-tuple of  $t_2$, we search for a match in the tree $\mathcal{T}$.
If we do not find any, we insert all of $t_2$'s $r$-tuples into $\mathcal{T}$ and we add $t_2$ to the output.
Otherwise, $t_2$ is not included in the output and we leave the tree unchanged.
We repeat this procedure for the remaining tracks.
In the end, the output contains all the desired tracks.
See Figure \ref{fig:redblacktree} for an illustration of the algorithm.

With each accepted track only $\binom{L}{r}=O(1)$ elements are inserted in $\mathcal{T}$.
Since $k_{\text{cand}}$ candidate tracks reach the cleaning stage, the size of the tree never exceeds $O(k_{\text{cand}})$.
So, we guarantee $O(\log k_{\text{cand}})$ complexity for the search and insertion tasks.
This means that we only spend $O(\log k_{\text{cand}})$ time per candidate track.
Overall, our cleaning algorithm has complexity $O(k_{\text{cand}} \log k_{\text{cand}})$.

To generalize this to the case of varied number of hits per track, note that we can only find up to $L = O(1)$ different track sizes.
Let $R = \lceil f L \rceil$.
We initialize $R^2$ empty balanced binary search trees $\mathcal{T}_{i,j}$ for $i, j \in \{1, 2, \ldots, R\}$.
The first track $t_1$ is immediately included in the output.
Say it has $L_1$ hits and let $r_1 = \lceil f L_1 \rceil$.
We insert all of the $r$-tuples of $t_1$ for $r \leq r_1$ into $\mathcal{T}_{r, r_1}$.
Let the second track $t_2$ have size $L_2$ and $r_2 = \lceil f m_2 \rceil$.
There are two cases to consider when two tracks share more than $\min(r_1, r_2)$ hits:
\begin{enumerate}[label=(\alph*)]
    \item $r_1 \leq r_2$: the overlapping tuples are represented in $\mathcal{T}_{r_1, r_1}$. Searching all trees $\mathcal{T}_{r, r}$ for $r \leq r_2$ will reveal the overlap.
    \item $r_1 > r_2$: the overlapping tuples are represented in $\mathcal{T}_{r_2, r_1}$. Searching all trees $\mathcal{T}_{r_2, r}$ for $r>r_2$ will reveal the overlap.
\end{enumerate}
If we do not find any match, we insert all of the $r$-tuples of $t_2$ for $r \leq r_2$ into $\mathcal{T}_{r, r_2}$ and add $t_2$ to the output.
Repeating this for all tracks will guarantee that there are no two tracks $t_i$ and $t_j$ in the output sharing more than $\min (r_i, r_j)$ hits.
We write down our improved version of cleaning in Algorithm \ref{algo:improvedcleaning}.

\begin{proofT}{Theorem \ref{thm:newcleaning}}
The reasoning is essentially the same as for constant-sized tracks.
For each accepted candidate track the number of tuples inserted into the corresponding search tree s bounded by $\binom{L}{r}=O(1)$.
Therefore, no tree will contain more than $O(k_{\text{cand}})$ elements, and the search and insert operations can always be performed in $O(\log k_{\text{cand}})$ time.
Since there are $R^2=O(1)$ trees, we spend $O(\log k_{\text{cand}})$ per track candidate.
\end{proofT}

\section{More on reconstructing tracks in superposition \label{app:reco_supersposition}}

We start by slightly adjusting the steps in the classical CTF track building algorithm (Algorithm \ref{algo:finding}). For a given seed, CTF selects up to $\lambda$ track candidates in each layer to propagate to the next layer. If fewer track candidates have acceptable $\chi^2$ value, fewer than $\lambda$ track candidates are formed. Here we form exactly $\lambda$ new track candidates for every given track candidate. If there is at least one hit with $\chi_l^2(\mathbf{m}_{l,j}) < \chi_0^2$, we use $\lambda$ hits with the lowest $\chi^2$ values to build the new track candidates. If there is no such hit, we use one ghost hit and $\lambda - 1$ hits with the lowest $\chi^2$ values. We also build tracks for all triplets in the seeding layer. This would add substantial unnecessary work in the classical case, but does not add complexity if performed in quantum superposition. As in other sections, we can add placeholder hits to ensure that each layer has exactly $n$ hits and all tracks traverse through all $L$ layers. Thus at the end of the track finding phase we have exactly $\lambda^L n^c$ track candidates.

Based on this modified algorithm we construct a family of unitary transformations $\mathbf{U}_i$ that perform seeding, track finding, cleaning and selection in superposition with the following effect:

\begin{align}
\label{eq:inner}
    \mathbf{U}_i\ket{0}  = 
    &
    \sqrt{\frac{1-\epsilon}{\lambda^L n^c}}
    \Bigg(\sum_{j=0}^{k_i-1}\ket{\psi_j}\ket{-q_{L-1, j}} + \nonumber \\  
    & + \sum_{j=k_i}^{\lambda^Ln^c-1}\ket{\psi_j}\ket{+\infty} \Bigg) +
    \sqrt{\epsilon}\ket{\psi_\epsilon}\ket{+\infty}.
\end{align}

Here $\ket{\psi_0}, \ket{\psi_1}, \ldots, \ket{\psi_{k_i-1}}$ are the computational basis states encoding track candidates that 
\begin{enumerate}[label=(\alph*)]
    \item the classical CTF algorithm would output after the track building stage,
    \item do not share too many hits with the $i$ tracks already added to the output.
\end{enumerate}
$q_{L-1, j}$ is the quality score \eqref{eq:qualityscore} at the last layer of the track encoded in $\psi_j$. We consider $-q_{L-1, j}$ to formulate the task as a minimization problem. The computational basis states $\ket{\psi_{k_i}}, \ket{\psi_{k_i+1}}, \ldots,$ $\ket{\psi_{\lambda^L n^c-1}}$ encode some track candidates that do not pass (a) or (b).
Note that the tracks that were previously sampled belong to this set of states.
``$+\infty$'' is a large positive value, so the minimum finding gives answers only in the useful subset of the entangled computational basis states. $\ket{\psi_\epsilon}$ is some arbitrary quantum state, and $\epsilon$ is the \emph{error} probability of $\mathbf{U}_i$, i.e., the probability that the measurement of $\mathbf{U}_i\ket{0}$ would produce a result other than one of $\psi_0, \psi_1, \ldots, \psi_{\lambda^L n^c-1}$.

Next we show that such a family of unitary transformations can indeed be constructed. We first consider the track building (Lemma \ref{lem:qram_finding}) and cleaning (Lemma \ref{lem:qram_cleaning}) sub-procedures. 
Track building prepares an equal superposition over the $\lambda^L n^c$  track candidates with an additional arbitrary quantum state \eqref{eq:inner_finding} representing the error probability of the algorithm. Selecting the track candidates of the original CTF after track building step is subsumed by the selection stage (Theorem~\ref{lem:Uinner}).

\begin{lemma}[Track building in superposition]
\label{lem:qram_finding}
There exists a unitary transformation $\mathbf{U}_\mathrm{build}$ \eqref{eq:inner_finding} that performs track building in $\tilde{O}(\sqrt{n})$ time in superposition.

\end{lemma}

\begin{equation}
\label{eq:inner_finding}
    \mathbf{U}_\mathrm{build}\ket{0} = \sqrt{\frac{1-\epsilon}{\lambda^L n^c}}\sum_{j=0}^{\lambda^Ln^c-1}\ket{\psi_j} + \sqrt{\epsilon}\ket{\psi_\epsilon}
\end{equation}

\begin{proof}{}
Preparing an equal superposition over all the possible seed $c$-tuplets, i.e., seeding (implicit in $\mathbf{U}_\mathrm{build}$), can be done in $\tilde{O}(1)$ time with the Walsh-Hadamard transform. We have seen  that we can perform track building for one seed in $\tilde{O}(\sqrt{n})$ time with constant probability. However, both quantum minimum finding and its sub-procedure -- quantum exponential searching algorithm -- use measurements. While we cannot use measurements in our unitary transformations $\mathbf{U}_\mathrm{build}$, we can apply the principle of deferred measurements \cite{NielsenChuang}. Whenever the Algorithm \ref{algo:quantumminimumfinding} performs a measurement, we can instead perform CNOT operations on an ancillary register. When the algorithm conditions a quantum operation on a measurement result, we can perform a controlled operation with the ancillary register as the control. The probability ($1-\epsilon$) to get the desired result based on measurements during the procedure or by deferring the measurement is the same. We can replicate the randomness in the quantum exponential searching algorithm \cite{BoyerSearch} by conditioning operations on the equal superposition of the allowed values $\{0, 1, \ldots, m\}$, where $m$ is an arbitrary integer. By conditioning on the digits of the binary representation of these values, as in quantum counting \cite{quantumcounting}, we can ensure that we only need $O(m)$ such operations and the asymptotic computational complexity of quantum exponential searching algorithm remains unchanged (up to constant factors).

One issue with this approach is that both the number of iterations of the outer loop of the quantum minimum finding (Algorithm \ref{algo:quantumminimumfinding}) and the time complexity of its sub-procedure -- quantum exponential searching \cite{BoyerSearch} -- may be proportional to $\sqrt{n}$. In the classical algorithm, if the quantum exponential searching takes more time, we can limit the number of iterations of the outer loop to ensure running time $\tilde{O}(\sqrt{n})$. In the quantum circuit we need to account for the worst case number of iterations in the main loop and the worst-case running time in the sub-procedure. This requires more than $\tilde{O}(\sqrt{n})$ gates. However, we can set a hard limit on the number of iterations of the main loop. Since the expected size of the search space decreases by more than a half with each iteration of the outer loop, the expected number of iterations to reach the minimum is less than $\log{n}+1$. If we limit the number of iterations of the outer loop to $\gamma(\log{n} + 1)$ for some constant $\gamma$, then by Markov's inequality the probability that we have not reached the minimum is less than $1/\gamma$. Since it is still upper-bounded by a constant, the rest of the analysis does not change. The limit on the number of iterations also implies a limit on the number of measurements and the required number of ancillary registers to account for the measurements in the quantum procedure.

\end{proof}

\begin{lemma}[Cleaning in superposition]
\label{lem:qram_cleaning}
There exists a unitary transformation that runs in $\tilde{O}(1)$ time and marks the track candidates that do not share any $r$-tuple of hits with any track already added to the output.
\end{lemma}

\begin{proof}{}
We have assumed that all particles traverse all layers, so all tracks are of length $L$ and are allowed to share up to exactly $r$ hits for some value of $r$. 
As with Algorithm \ref{algo:improvedcleaning}, we can generalize it to variable length tracks with a constant factor increase in complexity.
Like in the classical case, we can test whether an $r$-tuple has already been added to the tree $\mathcal{T}$ with $O(\log{n})$ queries to QRAM storing the values of the nodes of tree $\mathcal{T}$. 
Thus each track can be associated with a list of $\binom{L}{r}$ binary values indicating if an $r$-tuple has already been added to the output in $\tilde{O}(1)$ time.
Testing whether any of these values is equal to $1$ requires $O(1)$ gates.
Hence the total time required to mark the necessary track candidates is $\tilde{O}(1)$.
\end{proof}

\begin{lemma}
\label{lem:Uinner}
Each unitary transformation $\mathbf{U}_i$ \eqref{eq:inner} can be built to run in time $\tilde{O}(\sqrt{n})$.
\end{lemma}

\begin{proof}{}
We already saw that track building requires $\tilde{O}(\sqrt{n})$ time (Lemma \ref{lem:qram_finding}) and testing whether a track overlaps with any already added to the output takes $\tilde{O}(1)$ time (Lemma \ref{lem:qram_cleaning}). Procedures necessary for the track selection -- marking the tracks that pass the track building stage in CTF, refitting, recalculating the score and comparing to a threshold value -- depend on a constant number of fixed-precision numbers, and hence can be done in $\tilde{O}(1)$ time. So the time complexity of $\mathbf{U}_i$ is dominated by the track building and is $\tilde{O}(\sqrt{n})$.
\end{proof}

We will now describe how we can use transformations $\mathbf{U}_i$ with quantum minimum finding to reconstruct the tracks one-by-one (Algorithm \ref{algo:reco_in_superposition}). We will search for $\psi_i^* = \arg\min Q_i(\psi)$, where $Q_i(\psi)$ is the score encoded in the second register of $\mathbf{U}_i\ket{0}$ \eqref{eq:inner}. The time complexity of the quantum minimum finding \cite{DurrMinimum} remains the same (up to constant factors) if instead of quantum exponential searching \cite{BoyerSearch} we use amplitude amplification (Theorem \ref{thm:aa}).

\begin{theorem}[Amplitude amplification \cite{brassard2002quantum}]
\label{thm:aa}
Let $\mathcal{A}$ be any quantum algorithm that uses no measurements, and let $a$ denote the initial success probability of $\mathcal{A}$. There exists a quantum algorithm that finds a good solution using an expected number of applications of $\mathcal{A}$ and $\mathcal{A}^{-1}$ which are in $\Theta(1/\sqrt{a})$ if $a > 0$, and otherwise runs forever.
\end{theorem}

Let $a$ be the probability to find $\psi_i^*$ (or any specific track encoded in $\{\psi_0, \psi_1, \ldots, \psi_{k_i-1}\}$) by measuring $\mathbf{U}_i\ket{0}$. Then $a = {(1-\epsilon)/(\lambda^L n^c)}$ and the expected number of calls to $\mathbf{U}_i$ by the quantum minimum finding algorithm for a constant probability of error is 
\begin{equation}
    O\left(\sqrt{(\lambda^L n^c)/(1-\epsilon)}\right)=O\left(\sqrt{n^c}\right).
\end{equation}

As with Algorithm \ref{algo:quantumfinding}, repeating the quantum minimum finding algorithm $O(\log{n})$ times allows us to reduce the error probability to $O(1/n)$ so that sampling $O(n)$ tracks has a constant probability of error. In particular, since one application of the quantum minimum finding algorithm ensures failure probability smaller than $1/2$ and there cannot be more than $\lambda n^c$ tracks, the probability to not find the minimum in any of the iterations (if there are any valid tracks remaining) after repeating the algorithm $\lceil \log{2 \lambda n^c} \rceil$ times is below $1/2$.

Once we have found the best track, we can build a self-balancing binary tree $\mathcal{T}$ to be used in track selection for the next track. More generally -- suppose that we have found the $j$ best tracks tracks that the CTF algorithm outputs. Each time we find a new track, we insert it in $\mathcal{T}$. This tree never exceeds $O(n)$ size, and so the insertion operation cost is $O(\log{n})$. $\mathbf{U}_i$ queries $\mathcal{T}$ to mark as invalid those tracks that have an overlap with the $i$ tracks already added to output.

\begin{proofT}{Theorem \ref{thm:ctf}}
Each iteration of the main loop in Algorithm \ref{algo:reco_in_superposition} takes $\tilde{O}(\sqrt{n^c}\max_i T_{\mathbf{U}_i})$ time, where $\max_i T_{\mathbf{U}_i} = \tilde{O}(\sqrt{n})$ (Lemma \ref{lem:Uinner}). There are $O(n)$ iterations to reconstruct $O(n)$ tracks. Thus the total time complexity of Algorithm \ref{algo:reco_in_superposition} is
\begin{equation}
    \tilde{O}\left(n \cdot \sqrt{n^c} \cdot \sqrt{n}\right) = \tilde{O}\left(n^{\frac{c + 3}{2}}\right).
\end{equation}
\end{proofT}

We note that for the special case where tracks are not allowed to share any hits, the approach described in this section allows the complete removal of the cleaning stage. Once the best track is found, all the points that belong to it can be masked (removed) and the algorithm is run again to find the best track on the remaining points.

\newpage

\section{Pseudo-codes \label{app:pseudocodes}}

\begin{algorithm}
\SetArgSty{}
\caption{Seeding (classical)}
\label{algo:seeding}
\SetKwData{seed}{seed}
\SetKwInOut{Input}{input}\SetKwInOut{Output}{output}
\SetKwData{seedlist}{seed\_list}
\SetKwData{seed}{seed}
\SetKwData{newseed}{new\_seed}

\Input{event record}
\Output{seeds}

\seedlist $\leftarrow \{\mathbf{m}_{0,0}, \mathbf{m}_{0,1}, \ldots, \mathbf{m}_{0,n-1}\}$\;
\ForEach{layer $l$ from $1$ to $c-1$}{
    \ForEach{\seed in \seedlist}{
        \ForEach{hit $\mathbf{m}_{l,j}$ in layer $l$}{
        	\If{$\mathbf{m}_{l,j}$ is a valid continuation for \seed}{
        		\newseed $\leftarrow$ \seed $\cup$ $\mathbf{m}_{l,j}$\;
        		add \newseed to \seedlist\;
            	}
        }
        remove \seed from \seedlist\;
    }
}
output \seedlist\;

\end{algorithm}

\begin{algorithm}
\SetArgSty{}
\caption{Seeding with quantum search}
\label{algo:quantumseeding}
\SetKwInOut{Input}{input}\SetKwInOut{Output}{output}
\SetArgSty{}
\SetKwProg{Rp}{repeat}{}{many times}

\Input{event record}
\Output{seeds}

$\tilde{k} \leftarrow$ quantum counting estimation of  $k_{\text{seed}}$\label{step:counting}\;
$m \leftarrow \left\lfloor \pi \bigg/ 4 \arcsin \left( \sqrt{ \frac{k_{\text{seed}}}{n^c}  }\right) \right\rfloor$\;
\While{we have not found $\tilde{k}$ good seeds \label{step:seedwhile}}{
    prepare and measure state\label{step:groverstate} \begin{equation*}
    \mathbf{G}^m \cdot \left( \frac{1}{\sqrt{n^c}} \sum_{j_0, \ldots, j_{c-1} = 0}^{n-1}  \ket{ j_0, \ldots, j_{c-1} } \right);
    \end{equation*}
    
    \If{outcome $j_0, \ldots, j_{c-1}$ corresponds to a good seed}{
    add $(\mathbf{m}_{0,j_0}, \mathbf{m}_{1,j_1}, \mathbf{m}_{2,j_2})$ to output\label{step:seedoutput}\;
    }
}
\end{algorithm}

\begin{algorithm}
\SetArgSty{}
\caption{Track building (classical)}
\label{algo:finding}
\SetKwInOut{Input}{input}\SetKwInOut{Output}{output}
\SetKwData{seed}{seed}
\SetKwData{tracklist}{candidate\_tracks}
\SetKwData{track}{track}
\SetKwData{newtrack}{new\_track}

\Input{seeds, generated by Algorithm \ref{algo:seeding}; event record}
\Output{candidate tracks}

\ForEach{\seed}{
initialize empty list \tracklist\;
estimate initial state vector $\mathbf{p}_{c-1|c-1}$ and quality factor $q_{c-1}$ for seed\;
add (\seed, $\mathbf{p}_{c-1|c-1}$, $q_{c-1}$) to \tracklist\;
\ForEach{layer $l$ from $c$ to $L-1$}{
	\ForEach{(\track, $\mathbf{p}_{l-1|l-1}$, $q_{l-1}$) in \tracklist}{
	    evaluate predicted measurement $\mathbf{m}_{l|l-1}$\;
		\ForEach{hit $\mathbf{m}_{l,j}$ in layer $l$ \label{step:hitloop}}{
			\If{$\chi_{l|l-1}^2(\mathbf{m}_{l,j}) < \chi_0^2$ }{
			\newtrack $\leftarrow$ \track   $+ \ \mathbf{m}_{l,j}$\;
			form new candidate track for seed with $\mathbf{m}_{l,j}$\;
			evaluate $\mathbf{p}_{l|l}$ and quality factor $q_l$ for \newtrack\;
			add (\newtrack, $\mathbf{p}_{l|l}$, $q_l$) to \tracklist\;
			}
		}
		\If{there is no hit $\mathbf{m}_{l,j}$ in layer $l$ such that $\chi_{l|l-1}^2(\mathbf{m}_{l,j}) < \chi_0^2$}{
			\newtrack $\leftarrow$ \track   $+ \ \mathbf{m}_{l|l-1}$
			evaluate $\mathbf{p}_{l|l}$ and quality factor $q_l$ for \newtrack\;
			add (\newtrack, $\mathbf{p}_{l|l}$, $q_l$) to \tracklist\;
		}
	    remove (\track, $\mathbf{p}_{l-1|l-1}$, $q_{l-1}$) from \tracklist\;
	}
select the best $\lambda$ tracks of \tracklist
}
add elements of \tracklist to output\;
\tcp*[h]{note: this description uses notation from Appendix \ref{app:trackbuilding}}
}
\end{algorithm}

\begin{algorithm}
\SetArgSty{}
\caption{Quantum minimum finding}
\label{algo:quantumminimumfinding}
\SetKwInOut{Input}{input}\SetKwInOut{Output}{output}
\SetKwData{ept}{empty}

\Input{prediction of track's state vector at layer $l-1$}
\Output{$j$ such that $\mathbf{m}_{l,j}$ that minimizes $\chi^2_{l|l-1}$}

initialize $j_0 \leftarrow$ \ept\;
set $y \leftarrow \chi^2_0$\;
\While{$\mathbf{O}_{find}$ has been called less than $22.5 \sqrt{n}$ times}{
    apply quantum exponential searching algorithm of \cite{BoyerSearch} with initial state $\left(\frac{1}{\sqrt{n}} \sum_{j=0}^n \ket{l,j} \right) \cdot \ket{y}$ and with $\mathbf{O}_{find}$ as oracle\;
    \If{we find an state $\ket{l,j}$ such that $\chi^2_{l|l-1}(\mathbf{m}_{l,j}) < y $}{
        set $j_0 \leftarrow j$\;
        set $y \leftarrow \chi^2_{l|l-1}(\mathbf{m}_{l,j})$\;
    }
}
\eIf{$j_0$ is not \ept}{
    \KwRet{$\mathbf{m}_{l,j_0}$}
}{
    \KwRet{``no good measurement''}
}

\end{algorithm}

\begin{algorithm}
\SetArgSty{}
\caption{Track building with quantum minimum finding}
\label{algo:quantumfinding}
\SetKwInOut{Input}{input}\SetKwInOut{Output}{output}
\SetKwData{seed}{seed}
\SetKwData{tracklist}{candidate\_tracks}
\SetKwData{track}{track}
\SetKwData{newtrack}{new\_track}

\Input{seeds, generated by Algorithm \ref{algo:seeding}; event record}
\Output{candidate tracks}

\ForEach{\seed}{
initialize empty list \tracklist\;
estimate initial state vector $\mathbf{p}_{c-1|c-1}$ and quality factor $q_{c-1}$ for seed\;
add (\seed, $\mathbf{p}_{c-1|c-1}$, $q_{c-1}$) to \tracklist\;
\ForEach{layer $l$ from $c$ to $L-1$}{
	\ForEach{(\track, $\mathbf{p}_{l-1|l-1}$, $q_{l-1}$) in \tracklist}{
	    evaluate predicted measurement $\mathbf{m}_{l|l-1}$\;
	    \For{$i$ from $1$ to $\lambda$}{
    	    run quantum minimum finding (Algorithm \ref{algo:quantumminimumfinding}) $\log(3 L \lambda^2 n^c)$ times (increasing the $\chi^2$ value of already used hits so not to find them again) \label{step:qmflogtimes}\;
    	    from the samples of step \ref{step:qmflogtimes}, select the measurement $\mathbf{m}_{l,j}$ with lowest $\chi_{l|l-1}^2$ \label{step:measurementselect}\;
    	    \If{$\chi_{l|l-1}^2(\mathbf{m}_{l,j}) < \chi_0^2$ }{
    			\newtrack $\leftarrow$ \track   $+ \ \mathbf{m}_{l,j}$\;
    			form new candidate track for seed with $\mathbf{m}_{l,j}$\;
    			evaluate $\mathbf{p}_{l|l}$ and quality factor $q_l$ for \newtrack\;
    			add (\newtrack, $\mathbf{p}_{l|l}$, $q_l$) to \tracklist\;
    		}
		}
		\If{no new candidate track was formed}{
			\newtrack $\leftarrow$ \track   $+ \ \mathbf{m}_{l|l-1}$\;
			evaluate $\mathbf{p}_{l|l}$ and quality factor $q_l$ for \newtrack\;
			add (\newtrack, $\mathbf{p}_{l|l}$, $q_l$) to \tracklist\;
		}
	    remove (\track, $\mathbf{p}_{l-1|l-1}$, $q_{l-1}$) from \tracklist\;
	}
select the best $\lambda$ tracks of \tracklist\;
}
add elements of \tracklist to output\;
}
\end{algorithm}

\begin{algorithm}
\SetArgSty{}
\caption{Cleaning (original)}
\label{algo:cleaning}
\SetKwInOut{Input}{input}\SetKwInOut{Output}{output}
\SetKwData{tracka}{track$_1$}
\SetKwData{trackb}{track$_2$}

\Input{candidate tracks, generated by Algorithm \ref{algo:finding}}
\Output{cleaned candidate tracks}

\ForEach{\tracka in candidate tracks}{
    \ForEach{\trackb (different from \tracka) in candidate tracks}{
        \If{\tracka and \trackb share more than allowed fraction of hits}{
            remove the one with lowest quality score from the set of candidate tracks\;}
    }
}
output remaining candidate tracks\;
\end{algorithm}

\begin{algorithm}
\SetArgSty{}
\caption{Cleaning (improved)}
\label{algo:improvedcleaning}
\SetKwInOut{Input}{input}\SetKwInOut{Output}{output}
\SetKwData{tuple}{tuple}
\SetKwData{track}{track}

\Input{candidate tracks, generated by Algorithm \ref{algo:finding}}
\Output{cleaned candidate tracks}

sort candidate tracks by quality score\;
set $R = \lceil f L \rceil$\;
initialize empty trees $\mathcal{T}_{i,j}$ for $i,j \in \{1, \ldots, R\}$\;
\ForEach{\track in candidate track}{
    set $r = \lceil f L \rceil$, where $L$ is number of hits of \track\;
    \For{$r'$ from $1$ to $r$}{
        \ForEach{$r'$-\tuple of \track}{
            \If{$r'$-tuple is in $\mathcal{T}_{r',r'}$}{
                remove \track from set of candidate tracks\;
            }
        }
    }
    \For{$r'$ from $r$ to $R$}{
        \ForEach{$r$-\tuple of \track}{
            \If{$r$-tuple is in $\mathcal{T}_{r,r'}$}{
                remove \track from set of candidate tracks\;
            }
        }
    }
    \If{\track has not been removed}{
        \For{$r'$ from $1$ to $r$}{
            \ForEach{$r'$-\tuple of \track}{
                insert $r'$-\tuple into tree $\mathcal{T}_{r',r}$\;
            }
        }
    }
}
output remaining candidate tracks\;
\end{algorithm}

\begin{algorithm}
\SetArgSty{}
\caption{Selection (classical)}
\label{algo:selection}
\SetKwInOut{Input}{input}\SetKwInOut{Output}{output}
\SetKwData{track}{track}

\Input{candidate tracks, cleaned by Algorithm \ref{algo:cleaning}}
\Output{final reconstructed tracks}

\ForEach{\track in candidate tracks}{
    calculate quality score of \track\;
    \If{quality score of \track $<$ threshold}{
        remove \track from the set of candidate tracks\;}
}
output remaining candidate tracks\;
\end{algorithm}

\begin{algorithm}
\SetArgSty{}
\caption{Track reconstruction in superposition}
\label{algo:reco_in_superposition}
\SetKwInOut{Input}{input}\SetKwInOut{Output}{output}
\SetKw{stop}{stop}
\SetKwData{tuple}{tuple}
\SetKwData{tcj}{track$_j$}
\SetKwData{tci}{track}
\SetKwBlock{Rpt}{repeat}

\Input{event record}
\Output{reconstructed tracks}
initialize empty self-balancing binary tree $\mathcal{T}$\;
initialize $i \leftarrow$ 0\;

\Rpt{
    \For{$j$ from $1$ to $\lceil\log(2 \lambda n^c)\rceil$}{
            \tcj $\leftarrow$ output of the quantum minimum finding algorithm minimizing $Q_i(\psi)$ -- the score encoded in the second register of $\mathbf{U}_i\ket{0}$ \eqref{eq:inner}\;
    }
    \tci $\leftarrow \arg\min_j{Q_i(\tcj)}$\;
    \eIf{\tci passes CTF criteria and none of its $r$-tuples are in $\mathcal{T}$}
    {
        \ForEach{$r$-\tuple of \tci}{
            insert $r$-\tuple into tree $\mathcal{T}$\;
        }
        add \tci to output\;
        $i \leftarrow i + 1$\;
    }
    {
        \stop\;
    }
}
\tcp*[h]{note: this description uses notation from Appendix \ref{app:reco_supersposition}
}
\end{algorithm}

\end{document}